\documentclass[11pt]{article}
\usepackage{amsmath,amsthm,amsfonts,amscd,eucal,latexsym,amssymb}
\usepackage{hyperref}
\usepackage{epsfig}  
\usepackage{color}
\newcommand{\normal}[1]{\mathop{:}\nolimits\!#1\!\mathop{:}\nolimits} 
\newtheorem{theorem}{Theorem}[section]
\newtheorem{lemma}[theorem]{Lemma}
\newtheorem{proposition}[theorem]{Proposition}
\newtheorem{corollary}[theorem]{Corollary}
\theoremstyle{definition}
\newtheorem{definition}[theorem]{Definition}

\theoremstyle{remark}
\newtheorem{remark}{Remark}

\def\bC{{\mathbb C}}           
\def\bR{{\mathbb R}} 
\def\bH{{\mathbb H}}
\def\bZ{{\mathbb Z}}
\def\mA{{\mathcal A}}

\def\mO{{\mathcal O}}
\def\mS{{\mathcal S}}

\newcommand{\supp}{{\rm supp}}

\newcommand{\se}[1]{\section{#1}}
\newcommand{\sse}[1]{\subsection{#1}}

\begin{document}

\par 
\bigskip 
\par 
\rm 
 
\par 
\bigskip 
\LARGE 
\noindent 
{\bf The Casimir effect from the point of view of algebraic quantum field theory} 
\bigskip 
\par 
\rm 
\normalsize 

\large
\noindent {\bf Claudio Dappiaggi$^{1,2,a}$},
{\bf Gabriele Nosari$^{1,2,b}$}, {\bf Nicola Pinamonti}$^{3,4,c}$ \\
\par
\small
\noindent $^1$ 
Dipartimento di Fisica, Universit\`a degli Studi di Pavia,
Via Bassi, 6, I-27100 Pavia, Italy.\smallskip

\noindent$^2$  Istituto Nazionale di Fisica Nucleare - Sezione di Pavia,
Via Bassi, 6, I-27100 Pavia, Italy.\smallskip

\noindent$^3$ Dipartimento di Matematica, Universit\`a di Genova - Via Dodecaneso 35, I-16146 Genova, Italy.\smallskip

\noindent$^4$ Istituto Nazionale di Fisica Nucleare - Sezione di Genova, Via Dodecaneso, 33 I-16146 Genova, Italy.

\bigskip

\noindent $^a$  claudio.dappiaggi@unipv.it,
$^b$  gabriele.nosari@pv.infn.it, $^c$ pinamont@dima.unige.it
 \normalsize

\par 
 
\rm\normalsize 
\bigskip
\noindent {\small Version of \today}

\rm\normalsize 
 
 
\par 
\bigskip 

\noindent 
\small 
{\bf Abstract}. 
We consider a region of Minkowski spacetime bounded either by one or by two parallel, infinitely extended plates orthogonal to a spatial direction and a real Klein-Gordon field satisfying Dirichlet boundary conditions. We quantize these two systems within the algebraic approach to quantum field theory using the so-called functional formalism. As a first step we construct a suitable unital $*$-algebra of observables whose generating functionals are characterized by a labelling space which is at the same time optimal and separating and fulfils the F-locality property. Subsequently we give a definition for these systems of Hadamard states and we investigate explicit examples. In the case of a single plate, it turns out that one can build algebraic states via a pull-back of those on the whole Minkowski spacetime, moreover inheriting from them the Hadamard property. When we consider instead two plates, algebraic states can be put in correspondence with those on flat spacetime via the so-called method of images, which we translate to the algebraic setting. For a massless scalar field we show that this procedure works perfectly for a large class of quasi-free states including the Poincar\'e vacuum and KMS states. Eventually Wick polynomials are introduced. Contrary to the Minkowski case, the extended algebras, built in globally hyperbolic subregions can be collected in a global counterpart only after a suitable deformation which is expressed locally in terms of a *-isomorphism. As a last step, we construct explicitly the two-point function and the regularized energy density, showing, moreover, that the outcome is consistent with the standard results of the Casimir effect.
\normalsize
\bigskip 

\se{Introduction}
The success of quantum field theory (QFT) is often and righteously ascribed to the associated description of the matter constituents and of their mutual interactions. Yet one should not forget that QFT has lead to the discovery of several unique phenomena which could be thought as really being the blueprint of the theory. In between these one should certainly include the Casimir effect. Heuristically often depicted as the existence of a non vanishing force between two infinite, parallel and perfectly conducting plates due to the quantum fluctuations of the vacuum, it has been thoroughly studied in the literature. Much has been written about it since the publishing of the seminal papers \cite{Casimir:1948dh, CP} and also several confirmations at an experimental level are available. It is a daunting task to give an exhaustive bibliography of all different aspects of the Casimir effect and of all its ramifications in modern physics. We do not even pretend to trying and we limit ourselves in recommending an introduction which complements the content of this paper \cite{Milton}. 

Our goal is instead to fill a partial gap, namely to discuss the Casimir effect by means of the functional approach to the axiomatic description of quantum field theory. More precisely we are interested in the so called algebraic approach, a framework first introduced by Haag \& Kastler in the sixties -- see \cite{Haag:1963dh}, which divides the quantization of a physical system in two separate steps. The first consists of collecting all observables in a unital $*$-algebra whose mutual relations encode concepts such as locality and causality as well as information on the dynamics of the system. In the second step, one identifies a quantum state, that is a positive, normalized linear functional on the algebra of observables. Via the renown GNS theorem, one can recover the standard probabilistic interpretation of all quantum theories. 
The success of the so-called algebraic approach is doubtless, especially since it can be directly applied also to curved backgrounds, under minimal assumptions on the causal structure of the underlying background -- see for example \cite{Benini:2013fia, Hollands:2014eia}. With reference specifically to the Casimir effect, mathematically rigorous analyses can be found in \cite{Deutsch:1978sc, Kay:1978zr}, while preliminary investigations in the algebraic framework can be found in \cite{Kuhn, Ole, Sommer}. It is especially noteworthy the analysis in \cite{Sommer} which associates to a quantum field theory on a region with boundaries the universal algebra generated by the algebras of properly embedded globally hyperbolic subregions. Boundary conditions are take into account via appropriate ideals. Although this is a viable alternative, we shall not focus on it in this paper and we leave a comparison to our methods to future investigations.

Another notable and relevant exception is represented by \cite{Herdegen:2004xk, Herdegen:2005is, Herdegen:2010rt} although our approach should be seen as parallel and complementary rather than a continuation of these analyses. 

Let us now be more specific on the goals of this paper. Although the words ``Casimir effect" actually do encompass several rather different systems, we are interested only in two idealized scenarios, namely a real scalar field living in a bounded subset of Minkowski spacetime. The boundary is represented by either one or two parallel, infinitely extended hypersurfaces, orthogonal to a spatial direction. Dynamics is ruled by the Klein-Gordon equation with Dirichlet boundary conditions. The first case is dubbed {\em a Casimir-Polder system}, in analogy with \cite{CP}, the second {\em a Casimir system} with reference to \cite{Casimir:1948dh}. Our choices are merely for simplicity. On the one hand the method could be almost slavishly translated to Neumann or to Robin boundary conditions, although some of the conclusions, that we draw, such as the existence of KMS states, would not necessarily hold true. On the other hand the procedure we use could be similarly adapted to study other fields such as, for example, the vector potential. Yet, we feel that it is safer to start with a theory which does not include any gauge freedom which might complicate the analysis, hence with the risk of turning us from the main purpose of the paper. 

We give an answer to different questions concerning specific structural aspects of these systems. The first concerns which is the correct algebra of observables to associate to a free quantum field theory in a confined region such as those considered in a Casimir-Polder or in a Casimir system. This is not an obvious question since the standard procedure in the algebraic approach relies heavily on the underlying manifold being globally hyperbolic and on finding the smooth solutions to the equation(s) of motion, seen as an initial value problem. Both features are not present in our model. In order to tackle this problem we adapt to the case at hand the so-called functional formalism which has been used successfully in the algebraic framework in the past few years -- see for an introduction \cite{Brunetti:2009qc, Brunetti:2009pn, Fredenhagen:2012sb}. The net advantage of this procedure is the following: Observables are seen as functionals on a space of kinematical/dynamical configurations and the algebraic structure is obtained by deforming the standard pointwise product so to include the information of the canonical commutation relations. As soon as one wants to deal with interactions at a perturbative level or is interested in the expectation value of quantities such as the stress-energy tensor, Wick polynomials are needed. Although their rigorous construction is known since more than a decade \cite{Hollands:2001nf}, the functional formalism allows for an easier identification not only of the polynomials themselves but also of the underlying algebraic structure via an additional deformation of the pointwise product.

In order to select a specific class of functionals we adapt to the case at hand a procedure which was already successfully applied recently to the analysis of Abelian gauge theories \cite{Benini:2012vi, Benini:2013tra, Sanders:2012sf} and of linearized gravity \cite{Benini:2014rya}: We start by constructing the space of all possible configurations allowed by the underlying dynamics, by means of a well-known procedure in PDE theory: the method of images. Subsequently we identify a set of linear functionals on the collection of dynamical configurations which play the role of the above mentioned generators. In order to justify our choice we will argue that there are minimal requirements which need to be met, namely these generating functionals should be a separating and optimal set -- see \cite{Benini}. At this stage the analysis of a Casimir-Polder and of a Casimir system will start to diverge considerably. While in the first case we will show that generators are, up to an isomorphism, a subset of those for a Klein-Gordon field in Minkowski spacetime, in the second, this feature is lost. Additionally we verify that the algebra of observables also enjoys the so-called {\em F-locality property} introduce by Kay in \cite{Kay:1992es}, according to which, the restriction of such algebra to any globally hyperbolic subregion of the underlying manifold should be $*$-isomorphic to the algebra of observables built directly on this region with the standard prescriptions. An important novel point, which our investigation shall uncover, is that the algebra of observables both for a Casimir and a Casimir-Polder system enjoys the same structural properties of the standard Minkowski counterpart, especially the time-slice axiom, a feature which was not considered before. 

The second question to which we wish to give an answer concerns the choice of an algebraic quantum state of Hadamard form both for a Casimir and for a Casimir-Polder system. The microlocal characterization of the Hadamard condition was formulated by \cite{Radzikowski:1996pa, Radzikowski:1996ei} for scalar field theories on globally hyperbolic spacetimes. Here we extend the definition so that it can be applied also to theories in bounded regions. In particular we shall call a state Hadamard if such property is satisfied by its restriction to any globally hyperbolic submanifold of the underlying spacetime, extending at a level of states the above mentioned F-locality property. Subsequently we investigate explicit examples. Also at this stage, the two systems, that we consider, differ greatly. In the Casimir-Polder one, it turns out that algebraic states can be constructed via pull-back from those in the whole Minkowski spacetime inheriting, moreover, the Hadamard property. In the Casimir one, the situation is far more complicated. Here our main goal is to make contact with the procedures often followed in the standard physics literature, where states are constructed either with the method of Green functions or, exploiting the special geometry of the system, via the method of images -- for a preliminary investigation see \cite{Ole}. The aim especially of the latter is to show that one can construct states for a Casimir system starting directly from those for a Klein-Gordon field on the whole Minkowski spacetime. We stress one additional advantage, which is almost never mentioned: The method of images does not rely on modes and hence on a Fourier transform, being thus a natural candidate to be used for a generalization of our results to curved backgrounds. 
 We investigate how to translate rigorously this procedure in the algebraic framework and we show that, in the case of a massless real scalar field, if we start from the Poincar\'e vacuum, we obtain a full-fledged Hadamard state for a Casimir system. At the same time we show that we can consider a larger class of states on the whole Minkowski spacetime as starting point. More precisely we give sufficient conditions to identify them and we show that KMS states at finite temperature meet them. As a byproduct, it turns out that the corresponding state for the Casimir system preserves the KMS condition.

Additionally, in view of the microlocal characterization of the Hadamard states for the two systems under investigation, we are able to construct the extended algebra of Wick polynomials. Noteworthy is the fact that, in order to embed the local Wick polynomial, {\it i.e.}, those constructed in a globally hyperbolic subregion, into a global extended algebra, a non-local deformation of the ordinary star product is necessary. In this respect, we recall that the local extended algebra depends only on the choice of the Hadamard function, used to deform the star-product. Different choices of Hadamard functions yield isomorphic algebras and the intertwining isomorphism is a regular deformation \cite{Brunetti:2009qc}. 
Yet, in the systems under investigation, contrary to the Minkowski case, it is impossible to construct a global Hadamard function which depends only on local geometric properties. Hence, the Wick polynomial constructed out of local property of the spacetime can be represented in a global algebra only after applying a local deformation. The necessity for such deformation becomes manifest in the computation of the correlations between local observables constructed on suitable different globally hyperbolic subregions.

The synopsis of the paper is the following: We define notations and conventions in the next subsection. In the second section, instead we focus on a Casimir-Polder system. To start with, we classify all dynamically allowed configurations, constructing out of them the $*$-algebra of fields and relating it to a subalgebra of the one for a Klein-Gordon field on Minkowski spacetime. 
Subsequently we give a notion of Hadamard states for a Casimir-Polder system and we show how they are related to those on the whole Minkowski spacetime. Eventually we discuss the notion of Wick polynomials and of Hadamard regularization pointing out the differences with the standard approach. We show how one can recover, starting from the Poincar\'e vacuum, the usual results for the two-point function and for the regularized energy density. In the third section instead we focus on a Casimir system. Mimicking the same procedure of the second section, first we construct all dynamical configurations and then the unital $*$-algebra of fields. After giving the notion of Hadamard states, we investigate how to construct them starting from those on the whole Minkowski spacetime. In particular we discuss the method of images and we show that it gives well-defined results if we start either from the Poincar\'e vacuum or from a KMS state at finite temperature, if we consider a massless Klein-Gordon field. In this respect we extend to the algebraic framework earlier analyses, see in particular \cite{Brown:1969na, Fulling, Kennedy:1979ar}  for the thermal case and \cite{Fulling:1989nb} for the vacuum case. Eventually we compute also in this case the expectation value of the two-point function and of the regularized energy density.

\sse{Notations and conventions}\label{notations}

Goal of this section is to introduce the notation which will be used throughout the text.

Throughout this paper we shall always indicate with $\bR^4$ Minkowski spacetime, hence thinking of this manifold as endowed with the flat metric of signature $\mathrm{diag}(-1,1,1,1)$. We will consider always the standard Cartesian coordinates $(t,x,y,z)$ and we use $\underline{x}$ as a short-cut to indicate $(t,x,y)$. When needed, we will also write $\underline{x}^\mu$ to indicate the $\mu$-component, $\mu=0,1,2$. The remaining coordinate $z$ will play a distinguished role as we will be interested on subsets of $\bR^4$, bounded by one or two planes orthogonal to a spatial direction, which we will choose always as $z$. Additionally, in Section \ref{lastsec} we will be using bold face letters such as ${\bf x}$ as a short-cut to indicate $(x,y,z)$. 
From the point of view of quantum field theory, we shall think of $\bR^4$ as a special case of a globally hyperbolic spacetime (see \cite{BGP, Wald} for the definition). Furthermore we will be interested on the set of functions which are thereon smooth, $C^\infty(\bR^4)$, smooth and compactly supported, $C^\infty_0(\bR^4)$ or smooth and spacelike compact, $C^\infty_{sc}(\bR^4)$. The latter are those functions whose support, intersected with any Cauchy surface, is compact. Additionally we will need a fourth class of functions, which are not so often used in the literature and whose definition is, thus, here given for completeness:

\begin{definition}\label{timelikecompact}
We call $C^\infty_{tc}(\bR^4)$ the collection of {\em timelike compact} functions, that is those $\alpha\in C^\infty(\mathbb{R}^4)$ such that $\supp(\alpha)\cap J^-(p)$ and $\supp(\alpha)\cap J^+(p)$ is either compact or empty for all $p\in\mathbb{R}^4$. Here $J^\pm(p)$ indicate the causal future (+) and the causal past (-) of $p\in\bR^4$. 
\end{definition}

On Minkowski spacetime we shall consider a real scalar field whose dynamics is ruled by $P\phi\doteq\left(\Box-m^2\right)\phi=0$ where $\Box$ is the standard d'Alembert wave operator. As thoroughly discussed in \cite{BGP, Baernew}, since $P$ is a normally hyperbolic partial differential operator, there exist two operators $E^\pm: C^\infty_{tc}(\bR^4)\to C^\infty(\bR^4)$, called $E^+$, the advanced and $E^-$, the retarded fundamental solution, such that 
\begin{itemize}
\item $P\circ E^\pm$ and $E^\pm\circ P$ are the identity on $C^\infty_{tc}(\bR^4)$,
\item for all $\alpha\in C^\infty_{tc}(\bR^4)$, $\supp(E^\pm(\alpha))\subseteq J^\pm(\supp(\alpha))$.
\end{itemize}
Starting from $E^\pm$ we can build the causal propagator $E\doteq E^+-E^-$. It yields an isomorphism of topological vector spaces between $\frac{C^\infty_{tc}(\bR^4)}{P\left[C^\infty_{tc}(\bR^4)\right]}$ and $\mathcal{S}(\bR^4)=\left\{\phi\in C^\infty(\bR^4)\;|\; P\phi=0\right\}$ via $[\alpha]\mapsto E(\alpha)$. Notice that, if we consider $C^\infty_0(\bR^4)$ in place of $C^\infty_{tc}(\bR^4)$, then we obtain via the causal propagator all smooth and spacelike compact solutions to the equation $P\phi=0$.

\vskip .1cm

Another key role in our investigation is played by the algebra of observables for the scalar Klein-Gordon field in Minkowski spacetime. Following\cite{Brunetti:2009qc, Brunetti:2012ar, Fredenhagen:2012sb}, consider the space of kinematical/off-shell configurations $\mathcal{C}^{KG}(\bR^4)\equiv C^\infty(\bR^4)$, which is endowed with the compact-open topology. We consider regular functionals on the collection of kinematically allowed configurations \cite{Brunetti:2012ar}, namely

\begin{definition}\label{functional}
Let $F:\mathcal{C}^{KG}(\bR^4)\to\bC$ be any functional and let $\mathcal{U}\subset\mathcal{C}^{KG}(\bR^4)$ be an open set. We say that $F$ is {\em differentiable of order $k$} if, for all $m=1,...,k$, the following $m$-th order (G\^{a}teaux) derivatives exist as jointly continuous maps from $\mathcal{U}\times \left(\mathcal{C}^{KG}(\bR^4)\right)^{\otimes m}$ to $\bC$:
\begin{gather*}
F^{(m)}[\phi](\phi_1,...,\phi_m)=\langle F^{(m)}[\phi],\phi_1\otimes...\otimes\phi_m\rangle\doteq\left.\frac{\partial^m}{\partial\lambda_1...\partial\lambda_m}\right|_{\lambda_1=...=\lambda_m=0}F\left(\phi+\sum\limits_{j=1}^m\lambda_j\phi_j\right).
\end{gather*}
Here $\langle,\rangle$ denotes the dual pairing and, for each fixed $\phi\in\mathcal{C}^{KG}(\bR^4)$, $F^{(m)}[\phi]$ identifies a distribution density of compact support on $\bR^{4m}$. We say that a functional $F$ is 
\begin{itemize}
\item {\bf smooth} if it is differentiable at all orders $k\in\mathbb{N}$.
\item {\bf regular} if it is smooth, if, for all $k\geq 1$ and for all $\phi\in\mathcal{C}^{KG}(\bR^4)$, $F^{(k)}[\phi]\in C^\infty_0(\bR^{4k})$ and if only finitely many functional derivatives do not vanish.
We indicate this set as $\mathcal{F}_0(\bR^4)$.
\end{itemize}
\end{definition}

\noindent Since the dynamics is ruled by a Green-hyperbolic operator, we can endow $\mathcal{F}_0(\bR^4)$ with the structure of a $*$-algebra by means of the following product $\star:\mathcal{F}_0(\bR^4)\times\mathcal{F}_0(\bR^4)\to\mathcal{F}_0(\bR^4)$:
\begin{equation}\label{algprod}
\left( F\star F^\prime\right)(\phi)= \left(\mathcal{M}\circ\exp\left(i\Gamma_E\right)(F\otimes F^\prime)\right)(\phi),
\end{equation}
where $F,F^\prime\in\mathcal{F}_0(\bR^4)$. Here $\mathcal{M}$ stands for the pointwise multiplication, {\it i.e.}, $\mathcal{M}(F\otimes F^\prime)(\phi)\doteq F(\phi)F^\prime(\phi)$, whereas 
$$
\Gamma_{E}\doteq\frac{1}{2}\int_{\bR^4\times\bR^4}E(x,x^\prime)\frac{\delta}{\delta\phi(x)}\otimes\frac{\delta}{\delta\phi(x^\prime)},
$$
where $E(x,x^\prime)$ is the integral kernel of the causal propagator associated to $P$. The exponential in \eqref{algprod} is defined intrinsically in terms of the associated power series and, consequently, we can rewrite the product also as 
\begin{equation}\label{algprod2}
\left( F\star F^\prime\right)(\phi)=\sum\limits_{n=0}^\infty\frac{i^n}{2^n n!}\langle F^{(n)}(\phi),E^{\otimes n}(F^{\prime (n)})(\phi)\rangle,
\end{equation}
where the $0$-th order is defined as the pointwise multiplication, that is $\langle F^{(0)}(\phi),(F^{\prime (0)})(\phi)\rangle\doteq F(\phi)F^\prime(\phi)$. The $*$-operation is complex conjugation, that is, for all $F\in\mathcal{F}_0(\bR^4)$ and for all $\phi\in\mathcal{C}^{KG}(\bR^4)$, $F^*(\phi)=\overline{F(\overline{\phi})}$. We call 
$$
\mathcal{A}^{KG}(\bR^4)\doteq\left(\mathcal{F}_0(\bR^4),\star\right).
$$
While from a mathematical point of view it represents a deformation of the $*$-algebra of regular functionals endowed with the pointwise multiplication, from a physical point of view it describes an {\em off-shell algebra} of observables associated to the Klein-Gordon scalar field. Notice the following relevant facts:
\begin{itemize}
\item Since regular functionals are such that only a finite number of functional derivatives do not vanish, there is no issue concerning the convergence of \eqref{algprod2}. Furthermore we can realize $\mathcal{A}^{KG}(\bR^4)$ as being generated by functionals of the form 
\begin{equation}\label{generating}
\mathcal{F}_f(\phi)\doteq\int\limits_{\bR^4}d^4x\, f(x)\phi(x),\qquad f\in C^\infty_0(\bR^4),
\end{equation}
barring a completion needed to account for the fact that $C^\infty_0(\bR^4)\times...\times C^\infty_0(\bR^4)$ is dense in $C^\infty_0(\bR^4\times...\times\bR^4)$. In this respect smooth and compactly supported functions on Minkowski spacetime represent the labeling space of the off-shell algebra of functionals, building, thus, a bridge towards the more traditional approaches to a covariant quantization of a Klein-Gordon scalar field. 

\item Dynamics can be encoded simply restricting functionals to $\mathcal{S}(\bR^4)$, a vector subspace of $\mathcal{C}^{KG}(\bR^4)$ made of dynamically allowed configurations. As a by-product, $\mathcal{F}_0(\bR^4)$ contains redundant functionals, that is those $F\in\mathcal{F}_0(\bR^4)$ such that $F(\phi)=0$ for all $\phi\in\mathcal{S}(\bR^4)$. 
At the level of $\mathcal{A}^{KG}(\bR^4)$, this restriction can be implemented considering the quotient between such algebra and the ideal $\mathcal{I}^{KG}(\bR^4)$ generated by those functionals of the form \eqref{generating} with $f=P(h)$, $h\in C^\infty_0(\bR^4)$. Since the product $\star$ descends to the quotient, the result 
$$
\mathcal{A}^{KG}_{on}(\bR^4)\doteq\frac{\mathcal{A}^{KG}(\bR^4)}{\mathcal{I}^{KG}(\bR^4)},
$$ 
is again a $*$-algebra, which we dub also as the {\em on-shell algebra} of observables. Its labeling space is constituted by the equivalence classes lying in $\frac{C^\infty_0(\bR^4)}{P[C^\infty_0(\bR^4)]}$ and thus $\mathcal{A}^{KG}_{on}(\bR^4)$ is $*$-isomorphic to the algebra of observables built out of the standard approaches, see {\it e.g.} \cite{Benini:2013fia} and references therein.
\end{itemize}

\vskip .1cm

Before concluding this section we introduce a last notation which will be useful in this paper, namely we shall call $C^\infty_{-}(\bR^4)$ the set of all smooth functions on Minkowski spacetime such that $\alpha(\underline{x},z)=-\alpha(\underline{x},-z)$. We will also call the elements lying in this set as {\em odd} (under reflection along the hyperplane $z=0$). Conversely we refer to $C^\infty_{+}(\bR^4)$ as the collection of smooth functions which are {\em even} under reflection along the hyperplane $z=0$, that is $\alpha(\underline{x},z)=\alpha(\underline{x},-z)$. Notice the following splitting of vector spaces: $C^\infty(\bR^4)=C^\infty_-(\bR^4)\oplus C^\infty_+(\bR^4)$. Furthermore, since the operator $P$ contains only the second derivative along the $z$-direction, it holds that $P:C^\infty_{\pm}(\bR^4)\to C^\infty_{\pm}(\bR^4)$ and, thus, $P[C^\infty(\bR^4)]=P[C^\infty_-(\bR^4)]\oplus P[C^\infty_+(\bR^4)]$ as well as
\begin{equation}\label{split}
\frac{C^\infty(\bR^4)}{P[C^\infty(\bR^4)]}\simeq\frac{C^\infty_{-}(\bR^4)}{P[C^\infty_{-}(\bR^4)]}\oplus\frac{C^\infty_{+}(\bR^4)}{P[C^\infty_{+}(\bR^4)]}.
\end{equation}
The isomorphism \eqref{split} holds true even restricted to $C^\infty_{tc}(\bR^4)$, $C^\infty_{sc}(\bR^4)$ and to $C^\infty_0(\bR^4)$.


\se{Algebraic quantum field theory and the Casimir-Polder effect}\label{CPsystem}

Let us consider the following region of Minkowski spacetime, $(\mathbb{H}^4,\eta)$, where $\mathbb{H}^4=\mathbb{R}^3\times [0,\infty)$ is the four dimensional upper half-plane endowed with the Lorentzian flat metric. In agreement with the notations introduced in section \ref{notations}, the interval $[0,\infty)$ runs along the spatial direction, whose coordinate we indicate as $z$. We consider a real scalar field vanishing on the boundary $\partial\bH^4$ and whose dynamics is ruled by the Klein-Gordon equation. This scenario is often associated in the physics literature to the so-called Casimir-Polder effect \cite{CP}, which describes originally the interaction between a neutral atom in an electromagnetic cavity and a perfectly conducting wall at a distance $d$. 
For that reason, from now on we shall refer to our setting as a {\em Casimir-Polder system}.

We recall a standard definition in analysis \cite[Chapter 1]{Lee}:

\begin{definition}\label{smoothonbound}
Let $O\subseteq\bH^4$. We say that $u\in C^\infty(O)$ if and only if there exist both an open subset $\widetilde O$ of $\mathbb{R}^4$ such that $O\subseteq\widetilde O$ and  $\widetilde u\in C^\infty(\widetilde O)$ such that $\left.\widetilde u\right|_{O}=u$.
\end{definition}
Notice that the existence of $\widetilde u$ is guaranteed if and only if $u$ is continuous on the whole $O$, smooth on the interior $\mathring{O}\doteq O\setminus\partial O$ and each partial derivative of $u$ on $\mathring{O}$ has a continuous extension to $\partial O$. 
With this last definition and in full analogy with the standard case of a real scalar field on Minkowski spacetime, we call {\bf dynamical configurations} of a Casimir-Polder system the set $\mathcal{S}^{CP}(\bH^4)$ of all $u\in C^\infty(\mathbb{H}^4)$ such that $u$ satisfies the following boundary condition problem:  
\begin{equation}\label{dynCP}
\left\{\begin{array}{l}
P u = \left(\Box - m^2 -\xi R\right) u = 0, \qquad m\geq 0\,\textrm{and}\,\xi\in\bR\\
u(\underline{x},0)=0
\end{array}\right.,
\end{equation}
where $R$ is the scalar curvature. Although the scalar curvature on Minkowski spacetime or on any of its subsets vanishes identically, the coupling term $\xi R$ has a consequence on the form of the stress-energy tensor, which is proportional to the variation of the Lagrangian with respect to the metric.

Since $\bH^4$ is not a globally hyperbolic spacetime, we construct dynamical configurations via the method of images. The analysis which will involve the remainder of the section can be divided in three parts and it will follows conceptually the one outlined for the Klein-Gordon scalar field on the whole Minkowski spacetime in section \ref{notations}.
 
\vskip .2cm

\noindent\textbf{\em Part 1 -- Dynamical configurations:} Bearing in mind the notation introduced in section \ref{notations}, we introduce the isometry $\iota_z:\bR^4\to\bR^4$ for which $(\underline{x},z)\mapsto (\underline{x},-z)$. With a slight abuse of notation, we adopt the same symbol also to indicate its natural action on $C^\infty(\bR^4)$ and on generic distribution. We also recall, that, in view of Poincar\'e covariance, $\iota_z\circ E = E\circ \iota_z$, where $E$ is the causal propagator of $P$.
 
\begin{proposition}\label{CPdynchara}
Let $\mathcal{S}^{CP}(\bH^4)$ be the dynamical configurations of a Casimir-Polder system.  It holds that $\rho_{\bH^4}\circ\left(E-\iota_z\circ E\right):\frac{C^\infty_{tc,-}(\bR^4)}{P\left[C^\infty_{tc,-}(\bR^4)\right]}\to\mathcal{S}^{CP}(\bH^4)$ is a bijection. Here $\rho_{\bH^4}$ stands for the restriction map on $\bH^4$.
\end{proposition}

\begin{proof}
The map $E-\iota_z\circ E$ implements the {\em method of images} on Minkowski spacetime. Hence its image is a solution to Klein-Gordon equation on $\bR^4$ and, once restricted to $\bH^4$ via $\rho$, it implements also the Dirichlet boundary condition.  

Let us prove surjectivity. For any $u\in C^\infty(\bH^4)$ fulfilling \eqref{dynCP} we define the auxiliary function 
$$\widetilde u(\underline{x},z)=\left\{\begin{array}{ll}
u(\underline{x},z), & \forall\, (\underline{x},z)\in\bH^4\\
-u(\underline{x},-z) &\forall\, (\underline{x},z)\in \mathbb{R}^3\times (-\infty,0)
\end{array}\right. .
$$
Notice that $\widetilde u\in C^\infty(\bR^4)$. To show it, it suffices to control the behaviour of the function at $\partial\bH^4=\bR^3\times\{0\}$. Since $u(\underline{x},0)=0$ then $\widetilde u$ is continuous at $\partial\bH^4$. Let us consider now the first order partial derivatives: Continuity  at $\partial\bH^4$ is guaranteed along any of the $\underline{x}$-directions since $u(\underline{x},0)=0$ whereas that along the $z$-direction descends from the fact that $\widetilde u$ is odd along the $z$-directions and thus $\partial_z\widetilde u$ is even. A similar string of reasoning can be applied slavishly to the second derivative\footnote{We are grateful to Nicol\'o Drago, Igor Kahvkine and Valter Moretti for an enlightening discussion on this point.} barring for that along the $z$-direction for which we have first to recall that $\partial^2_z u(\underline{x},z)=\left(\partial^2_t-\partial^2_x-\partial^2_y + m^2+\xi R\right)u(\underline{x},z)$. Consequently since $u$ vanishes on $\partial\bH^4$, so does $\partial^2_z u$. Reiterating the procedure to all orders yields in combination with Schwarz theorem the sought result. Furthermore, since $u$ is a solution of \eqref{dynCP}, it holds that $P\widetilde u=0$. Consequently, in view of our discussion in section \ref{notations}, there exists $[\alpha]\in \frac{C^\infty_{tc}(\bR^4)}{P[C^\infty_{tc}(\bR^4)]}$ such that $\widetilde u = E(\alpha)$.

Since $\widetilde u$ is an odd function, $0=\widetilde u+\iota_z\widetilde u= E(\alpha)+\iota_z E(\alpha)=E(\alpha+\iota_z\alpha)=0$.  Hence there exists $\lambda\in C^\infty_{tc,+}(\bR^4)$, for which $\alpha+\iota_z\alpha = P\lambda$. If we add the information that $E\circ P=0$ and that $P[C^\infty_{tc}(\bR^4)]=P[C^\infty_{tc,+}(\bR^4)]\oplus P[C^\infty_{tc,-}(\bR^4)]$, to each $u\in\mS^{CP}(\bH^4)$, we can associate an equivalence class $[\alpha]\in\frac{C^\infty_{tc,-}(\bR^4)}{P\left[C^\infty_{tc,-}(\bR^4)\right]}$. This proves surjectivity. Notice that this map is per construction injective as, if $u=0$, then $\widetilde u=0$ and, thus we can write $\widetilde u=E(\alpha)$ with $\alpha\in P\left[C^\infty_{tc,-}(\bR^4)\right]$.
\end{proof}

\noindent\textbf{\em Part 2 -- The off-shell algebra:} Following the scheme given in Section \ref{notations}, we define at first a space of {\bf kinematical configurations} for a Casimir-Polder system. Let us introduce the following map:
\begin{equation}\label{eta-map}
\eta:C^\infty(\bR^4)\to C^\infty(\bH^4),\quad\phi(\underline{x},z)\mapsto u(\underline{x},z)\doteq\left.\frac{1}{\sqrt{2}}\left(\phi(\underline{x},z)-\iota_z\phi(\underline{x},z)\right)\right|_{\bH^4},
\end{equation}
where the numerical pre-factor is a normalization. 
\begin{definition}\label{CPoff}
We call {\em space of kinematical/off-shell configurations} for a Casimir-Polder system 
$$\mathcal{C}^{CP}(\bH^4)\doteq\left\{u\in C^\infty(\bH^4)\;|\;\left.u\right|_{\partial\bH^4}=0\;\textrm{and}\;\exists \phi\in \mathcal{C}^{KG}(\bR^4)\;\textrm{such that}\;u=\eta(\phi)\right\},$$
where $\eta$ is the map defined in \eqref{eta-map}. Since $\mathcal{C}^{CP}(\bH^4)\subset C^\infty(\bH^4)$ and since $\eta$ is for construction surjective thereon, we endow $\mathcal{C}^{CP}(\bH^4)$ with the quotient topology. In complete analogy we shall also consider $\mathcal{C}^{CP}_0(\bH^4)$ where the subscript $0$ stands for compact support. 
\end{definition}
\noindent For later convenience, we introduce $\eta^\dagger:C^\infty_0(\bH^4)\to \mathcal{E}^\prime(\bR^4)$ defined via
\begin{equation*}
\langle \eta^\dagger (h), \phi\rangle\doteq\langle h, \eta(\phi)\rangle_{\bH^4}=\int\limits_{\bH^4}d^4x\,h(x)\eta(\phi)(x),
\end{equation*}
Notice both that the integral kernel of $\eta^\dagger(h)\in\mathcal E^\prime(\bR^4)$ is $\frac{1}{\sqrt{2}}\left(h(x)\Theta(z)-[\iota_z(h)](x)\Theta(-z)\right)$ where $\Theta$ is the Heaviside step function and that
\begin{equation}\label{wfetadaga}
WF(\eta^\dagger(h))\subset\{(x, k)\in T^*\bR^4\setminus\{0\} \mid x\in\partial\bH^4 \,\,\text{and}\,\, k^i=g^{ij}k_j=0,\,\forall i\neq z\},
\end{equation}
where $g$ stands here for the Minkowski metric written in Cartesian coordinates.

In analogy with Definition \ref{functional}, we now introduce regular functionals on $\mathcal{C}^{CP}(\bH^4)$.
\begin{definition}\label{functionalCP}
Let $F:\mathcal{C}^{CP}(\bH^4)\to\bC$ be any smooth functional. We call it {\bf regular} if for all $k\geq 1$ and for all $u\in\mathcal{C}^{CP}(\bH^4)$, $F^{(k)}[u]\in C^{\infty}_0(\bH^{4k})$ and if only finitely many functional derivatives do not vanish. We indicate this set as $\mathcal{F}_0(\bH^4)$.
\end{definition}

\noindent In order to introduce a suitable product in $\mathcal{F}_0(\bH^4)$, analogous to \eqref{algprod}, we define a map which plays the role of $E$ in \eqref{algprod2}:
\begin{equation}\label{modified}
E_{\bH^4}: C^\infty_0(\bH^4)\to\mathcal S^{CP},\quad E_{\bH^4}(h)\doteq \eta\circ E\circ\eta^\dagger(h),
\end{equation}
and we call it {\bf CP-propagator}. Observe that $E\circ\eta^\dagger(h)$ is well-defined in view of \eqref{wfetadaga} and of \cite[Th. 8.2.13]{Hormander1}. Furthermore, for all $h\in C^\infty_0(\bH^4)$,  $E(\eta^\dagger(h))\in C^\infty_{sc}(\bR^4)$ and it solves the Klein-Gordon equation. As a by-product $E_{\bH^4}$ is well-defined map into $\mathcal S^{CP}$.

Let us consider now:
$$\star_{\bH^4}:\mathcal{F}_0(\bH^4)\times\mathcal{F}_0(\bH^4)\to\mathcal{F}_0(\bH^4),$$
which associates to each $F,F^\prime\in\mathcal{F}_0(\bH^4)$
\begin{equation}\label{starH}
\left(F\star_{\bH^4}F^\prime\right)(u)=\left(\mathcal{M}\circ\exp(i\Gamma_{E_\bH^4})(F\otimes F^\prime)\right)(u).
\end{equation}
Here $\mathcal{M}$ stands for the pointwise multiplication, {\it i.e.}, $\mathcal{M}(F\otimes F^\prime)(u)\doteq F(u)F^\prime(u)$, whereas 
$$\Gamma_{E_{\bH^4}}\doteq\frac{1}{2}\int_{\mathbb{H}^4\times\mathbb{H}^4}E_{\bH^4}(x,x^\prime)\frac{\delta}{\delta u(x)}\otimes\frac{\delta}{\delta u(x^\prime)},$$
where $E_{\bH^4}(x,x^\prime)$ is the integral kernel of \eqref{modified}. The exponential in \eqref{starH} is defined intrinsically in terms of the associated power series and, consequently, we can rewrite the product also as 
\begin{equation}\label{algprod3}
\left( F\star_{\bH^4} F^\prime\right)(u)=\sum\limits_{n=0}^\infty\frac{i^n}{2^n n!}\langle F^{(n)}(u),E_{\bH^4}^{\otimes n}(F^{\prime (n)})(u)\rangle_{\bH^4},
\end{equation}
where $\langle,\rangle_{\bH^4}$ stands for the pairing on $\bH^4$ built of out integration. The $0$-th order is defined as the pointwise multiplication, that is $\langle F^{(0)}(u),F^{\prime (0)}(u)\rangle\doteq F(u)F^\prime(u)$. Notice that \eqref{algprod3} defines a $\star$-product. In view of \eqref{modified}, $\langle F^{(n)}(u),E_{\bH^4}^{\otimes n}(F^{\prime (n)})(u)\rangle_{\bH^4}$ is well-defined for all $n\geq 0$ and the non-vanishing derivatives of $F\star_{\bH^4} F^\prime$ evaluated on any $u\in\mathcal{C}^{CP}(\bH^4)$ are compactly suppported.

\begin{definition}\label{PolCP}
We call $\mathcal{A}^{CP}(\bH^4)\equiv\left(\mathcal{F}_0(\bH^4),\star_{\bH^4}\right)$ the {\em off-shell $*$-algebra} of a Casimir-Polder system endowed with complex conjugation as $*$-operation. It is generated by the functionals $F_h(u)=\int\limits_{\bH^4}d^4x\, u(x)h(x)$ where $h\in C^\infty_0(\bH^4)$ while $u\in\mathcal{C}^{CP}(\bH^4)$. 
\end{definition}

\vskip .2cm

\noindent\textbf{\em Part 3 -- The on-shell algebra:} To conclude our investigation on the algebra of observables for a Casimir-Polder system, we want to investigate how $\mathcal{A}^{CP}(\bH^4)$ should be modified if we restrict the allowed configurations from $\mathcal{C}^{CP}(\bH^4)$ to $\mathcal{S}^{CP}(\bH^4)$. At this stage it is more advantageous to work on the counterpart of $\mathcal{S}^{CP}(\bH^4)$ on Minkowski spacetime specified by Proposition \ref{CPdynchara}.

\begin{proposition}\label{obsCP}
Let $\mathcal O^{KG}_-(\bR^4)$ be the span of all functionals $F_{[\zeta]}:\frac{C^\infty_{tc,-}(\bR^4)}{P[C^\infty_{tc,-}(\bR^4)]}\to\bC$, $[\zeta]\in \frac{C^\infty_{0,-}(\bR^4)}{P[C^\infty_{0,-}(\bR^4)]}$  such that $F_{[\zeta]}([\alpha])=\int\limits_{\bR^4}\,\zeta E(\alpha)$. This space is:
\begin{enumerate}
\item {\em\bf separating}, that is for every pair of different configurations $[\alpha],[\alpha^\prime]\in\frac{C^\infty_{tc,-}(\bR^4)}{P[C^\infty_{tc,-}(\bR^4)]}$, there exists a classical observable $[\zeta]\in\frac{C^\infty_{0,-}(\bR^4)}{P[C^\infty_{0,-}(\bR^4)]}$ such that $F_{[\zeta]}([\alpha])\neq F_{[\zeta]}([\alpha^\prime])$.
\item {\em\bf optimal}, that is, for every pair of classical observables $[\zeta],[\zeta^\prime]\in\frac{C^\infty_{0,-}(\bR^4)}{P[C^\infty_{0,-}(\bR^4)]}$ there exists at least one configuration $[\alpha]\in\frac{C^\infty_{tc,-}(\bR^4)}{P[C^\infty_{tc,-}(\bR^4)]}$ such that $F_{[\zeta]}([\alpha])=F_{[\zeta^\prime]}([\alpha])$
\item {\em\bf symplectic} if endowed with the following weakly non-degenerate symplectic form\footnote{Notice that, from a geometrical point of view, it would be more appropriate to refer to $\sigma$ as a Poisson structure. We stick to the more traditional codification used in algebraic quantum field theory.}:
$$
\sigma:\mathcal{O}^{KG}_-(\bR^4)\times\mathcal{O}^{KG}_-(\bR^4)\to\bR,\quad \sigma(F_[\zeta],F_[\zeta^\prime])=
\langle\zeta,E(\zeta^\prime)\rangle=\int\limits_{\bR^4}d^4x\; \zeta(x) E(\zeta^\prime)(x).
$$
\end{enumerate}
\end{proposition}

\begin{proof}
Let us prove {\em 1.} Consider any pair $[\alpha],[\alpha^\prime]\in\mathcal{S}^{CP}(\bH^4)$, $[\alpha]\neq [\alpha^\prime]$, and two representatives $\alpha,\alpha^\prime \in C^\infty_{tc,-}(\bR^4)$. On account of standard arguments in analysis we know that $C^\infty_{0}(\bR^4)$ is separating for $C^\infty(\bR^4)$ with respect to the pairing \eqref{functional}. Hence, since $E(\alpha-\alpha^\prime)\in C^\infty(\bR^4)$ is not vanishing, there must exist $\beta\in C^\infty_0(\bR^4)$ such that $(\beta,E(\alpha-\alpha^\prime))\neq 0$. Since $E(\alpha-\alpha^\prime)\in C^\infty_{-}(\bR^4)$, it holds that 
$(\beta,E(\alpha-\alpha^\prime))=(\zeta,E(\alpha-\alpha^\prime))$ where $\zeta(\underline{x},z)\doteq\beta(\underline{x},z)-\beta(\underline{x},-z)\in C^\infty_{0,-}(\bR^4)$. $\zeta$ identifies a non trivial element in $\mathcal{O}^{KG}_-(\bR^4)$, hence the statement is proven. 

We focus on {\em 2.} Let $[\zeta],[\zeta^\prime]\in\mathcal{O}^{KG}_-(\bR^4)$ and let $\zeta,\zeta^\prime$ be two arbitrary representatives. For the same reason as in the previous point, since $E(\zeta-\zeta^\prime)\in C^\infty(\bR^4)$ is non vanishing there must exist $\gamma\in C^\infty_0(\bR^4)$ such that $\supp(\gamma)\cap\left(\supp(E(\zeta))\cup\supp(E(\zeta^\prime))\right)\neq\emptyset$ and that $(\gamma,E(\zeta-\zeta^\prime))\neq 0$. Let $\alpha(\underline{x},z)\doteq\gamma(\underline{x},z)-\gamma(\underline{x},-z)\in C^\infty_{0,-}(\bR^4)\subset C^\infty_{tc,-}(\bR^4)$ individuate an element in $\mathcal{S}^{CP}(\bH^4)$ via the action of the causal propagator $E$. It holds that $F_{[\zeta]-[\zeta^\prime]}([\alpha])=(\zeta-\zeta^\prime, E(\alpha))=-(E(\zeta-\zeta^\prime),\alpha))=2(E(\zeta-\zeta^\prime),\gamma)\neq 0$, which entails the sought result.

At last we prove {\em 3.} Notice that, per construction, $\sigma$ is bilinear and antisymmetric. Suppose that, per absurd, there exists a non trivial $F_[\alpha]\in\mathcal{O}^{KG}_-(\bR^4)$ such that $\sigma(F_[\alpha],F_[\alpha^\prime])=0$ for every $F[\alpha^\prime]\in\mathcal{O}^{KG}_-(\bR^4)$. Since every representative of $[\alpha]$ is odd, the same statement holds true for every $[\alpha^\prime]\in\frac{C^\infty_0(\bR^4)}{P[C^\infty_0(\bR^4)]}$ since $\frac{C^\infty_{0,-}(\bR^4)}{P[C^\infty_{0,-}(\bR^4)]}$ and $\frac{C^\infty_{0,+}(\bR^4)}{P[C^\infty_{0,+}(\bR^4)]}$ are orthogonal to each other with respect to $\sigma$. 
\end{proof}

\begin{corollary}\label{crucialedirei}
Let $\mathcal{O}^{CP}(\bH^4)$ be the span of all functionals $F_{[f]}:\mathcal{S}^{CP}(\bH^4)\to\bC$ with $[f]\in\frac{\mathcal{C}^{CP}_0(\bH^4)}{P[\mathcal{C}^{CP}_0(\bH^4)]}$ such that $F_{[f]}(u)=\int\limits_{\bH^4}d^4x\; f(x)u(x)$, endowed with the symplectic form:
$$
\sigma_{\bH^4}:\mathcal{O}^{CP}(\bH^4)\times\mathcal{O}^{CP}(\bH^4)\to\bR,\quad \sigma_{\bH^4}(F_{[f]},F_{[f^\prime]})
\doteq \langle f, E_{\bH^4}f^\prime \rangle_{\bH^4} =\int\limits_{\bH^4}d^4x\; f(x) E_{\bH^4}(f^\prime)(x).
$$
There exists an isomorphism of symplectic spaces between $\mathcal{O}^{CP}(\bH^4)$ and $\mathcal{O}^{KG}_-(\bR^4)$.
\end{corollary}

\begin{proof}
First of all we notice that $F_{[f]}(u)$ with $u\in\mathcal{S}^{CP}(\bH^4)$ and $[f]\in\frac{\mathcal{C}^{CP}_0(\bH^4)}{P[\mathcal{C}^{CP}_0(\bH^4)]}$ is well-defined as the choice of the representative of $[f]$ is not relevant. As a matter of fact, on account of the boundary conditions of all elements involved, we can still integrate by parts canceling all boundary terms so that, for all $Pf^\prime$, $f^\prime\in\mathcal{C}^{CP}_0(\bH^4)$, $\int\limits_{\bH^4}d^4x\, P(f^\prime) u=\int\limits_{\bH^4}d^4x\,f^\prime Pu=0$ since $u\in\mathcal{S}^{CP}(\bH^4)$.
To prove the isomorphism we construct explicitly a bijective map from $\mathcal{O}^{CP}(\bH^4)$ to $\mathcal{O}^{KG}_-(\bR^4)$, preserving the symplectic structure.
We observe that the map $\eta$ of \eqref{eta-map} is injective on $C_{0,-}^\infty(\bR^4)$ thus it admits an inverse map $\eta^{-1}$ defined on $\eta[C_{0,-}^\infty(\bR^4)]\equiv\mathcal{C}^{CP}_0(\bH^4)$. Since both $\eta$ and $\eta^{-1}$ descend to the quotients $\frac{C^\infty_{0,-}(\bR^4)}{P[C^\infty_{0,-}(\bR^4)]}$ and $\frac{\mathcal{C}^{CP}_0(\bH^4)}{P[\mathcal{C}^{CP}_0(\bH^4)]}$, we can define with a slight abuse of notation the pull-back:
\begin{equation}\label{eta_star}
\eta^\ast:\mathcal{O}^{CP}(\bH^4) \to \mathcal{O}^{KG}_-(\bR^4) 
,\quad
\eta^\ast(F_{[f]})([\alpha])\doteq F_{[f]}(\eta(\phi)),  
\end{equation}
where $\phi=E([\alpha])\in\mathcal S^{KG}$.
Since any $u\in\mathcal S^{CP}$ is the unique image of a $[\alpha]\in \frac{C^\infty_{tc,-}(\bR^4)}{P[C^\infty_{tc,-}(\bR^4)]}$ via the bijection {\em 2.} of Proposition \ref{CPdynchara}, $\eta^*$ is an isomorphism of vector spaces. It also preserving the symplectic structure $\sigma_{\bH^4}$. To prove it, let us observe that $\eta^*(F_{[f]})=F_{\eta^\dagger([f])}$. We thus can write:
\begin{gather*}
\sigma(\eta^*(F_{[f]}), \eta^*(F_{[f^\prime]}))=\sigma(F_{\eta^\dagger([f])}, F_{\eta^\dagger([f^\prime])})\\
=\langle \eta^\dagger(f), E(\eta^\dagger(f^\prime))\rangle=\langle f,\eta E\eta^\dagger(f^\prime)\rangle_{\bH^4}=\\
=\sigma_{\bH^4}(F_{[f]}, F_{[f^\prime]}),
\end{gather*}
which is valid for all $F_{[f]}, F_{[f^\prime]}\in \mathcal O^{CP}(\bH^4)$.
\end{proof}

\noindent In view of this corollary, 

\begin{definition}\label{CPalgebra}
We call {\bf on-shell $*$-algebra of observables for a Casimir-Polder system} the algebra $\left(\mathcal{A}^{CP}_{on}(\bH^4),\star_{\bH^4}\right)$ generated by the functionals $\mathcal O^{CP}(\bH^4)$, where $\star_{\bH^4}$ is defined in \eqref{starH}.
\end{definition}

Before proving several important properties of $\mathcal{A}^{CP}_{on}(\bH^4)$, we want to investigate how it relates with the Minkowski counterpart $\mathcal A^{KG}(\bR^4)$. This will give us the chance to prove the already mentioned properties.

\begin{proposition}\label{globalcomparison}
Let $\widetilde\eta^*:\mathcal{A}_{on}^{CP}(\bH^4)\to\mathcal{A}_{on}^{KG}(\bR^4)$ be the natural extension of the pull-back map $\eta^*$ defined on $\mathcal O^{CP}(\bH^4)$. This is an injective $*$-homomorphism of algebras which becomes an isomorphism on $\mathcal{A}^{KG}_{on,-}(\bR^4)$, the $*$-subalgebra of $\mathcal{A}^{KG}(\bR^4)$ generated by functionals $\mathcal O^{KG}_-(\bR^4)$.
\end{proposition}

\begin{proof}
Let us prove that $\widetilde\eta^*$ is injective. Suppose that there exists $F_{[f]}\in\mathcal{A}_{on}^{CP}(\bH^4)$ such that $\eta^*(F_{[f]})$ is the vanishing functional. 
Then, for all $\phi\in\mathcal{S}^{KG}(\bR^4)$ one has $0=\eta^*(F_{[f]})(\phi)=F_{[f]}(\eta(\phi))=F_{\eta^\dagger([f])}(\phi)$.
Since $\phi$ is arbitrary, the only possible solution is $\eta^\dagger([f])=0$ and, thus $f=0$ for all $f\in[f]$, which entails the sought injectivity. 
In order to prove that $\widetilde\eta^*$ is also a $\ast$-homomorphism it suffices to focus again only on the generators.
Let $F_{[f]},F_{[f^\prime]}\in \mathcal O^{CP}(\bH^4)$ and $\phi\in\mathcal{S}^{KG}(\bR^4)$. Then, on account of \eqref{algprod2} the following holds true:
\begin{align*}
\left(\eta^*(F_{[f]})\star\eta^*(F_{[f^\prime]})\right)(\phi)&=(F_{\eta^\dagger([f])}\star F_{\eta^\dagger([f^\prime])})(\phi)\\
&=F_{\eta^\dagger([f])}(\phi)F_{\eta^\dagger([f^\prime])}(\phi)+\frac{i}{2}\langle \eta^\dagger(f), E(\eta^\dagger(f^\prime))\rangle=\\
&=F_{[f]}(\eta(\phi))F_{h^\prime}(\eta(\phi))+\frac{i}{2}\langle f,\eta E\eta^\dagger(f^\prime)\rangle_{\bH^4}=\\
&=\left(F_{[f]}\star_{\bH^4}F_{[f^\prime]}\right)[\phi].
\end{align*}
Since the $*$-operation is complex conjugation, it is left untouched by all the operations above and, as a consequence, we can infer that $\widetilde\eta^*$ is a $*$-homomorphism. 
The isomorphism $\mathcal{A}^{CP}_{on}(\bH^4)\simeq\mathcal{A}^{KG}_{on,-}(\bR^4)$ descends directly from Corollary \ref{crucialedirei}.
\end{proof}

\noindent In the following proposition, we investigate the structural properties of $\mA^{CP}(\bH^4)$, in particular causality and the time-slice axiom \cite{Brunetti:2001dx, Dimock}. The latter property needs a few comments. Recall that $\mathcal{A}^{KG}_{on}(\bR^4)$ fulfills the time-slice axiom, namely, given any open neighbourhood $\mathcal{N}$ of a Cauchy surface $\Sigma$ in Minkowski spacetime, containing all causal curves whose endpoints lie in $\mathcal{N}$, then $\mathcal{A}^{KG}_{on}(\bR^4)$ is $*$-isomorphic to $\mathcal{A}^{KG}_{on}(\mathcal{N})$. Since $\bH^4$ is not globally hyperbolic there is no notion of a Cauchy surface. Yet, if we consider the extension of the isomorphism of Proposition \ref{globalcomparison} to $\mathcal{A}^{CP}_{on}(\bH^4)$, this is $*$-isomorphic to a $*$-subalgebra of $\mathcal{A}^{KG}_{on}(\bR^4)$ for which the time-slice axiom is a well-defined concept. 
In addition to these two properties, we show that $\mathcal{A}^{CP}_{on}(\bH^4)$ satisfies the F-locality condition \cite{Fewster:1995bu, Kay:1992es}, a requirement which should be met by the algebra of observables of a quantum field theory on a non globally-hyperbolic spacetime. In a few words and in the case at hand, it requires that $\mathcal{A}^{CP}(\bH^4)$ and $\mathcal{A}^{KG}(\bR^4)$, restricted to any globally hyperbolic subregion of $\bH^4$ must be $*$-isomorphic. Such condition can be seen as an extension of the more renown local covariance, according to which, from local algebras, one should not be able to extract information on the global structure of the background\footnote{Recent experience with gauge theories teaches us that such conclusion should be read cum grano salis -- see for example \cite{Benini:2013tra,Sanders:2012sf}}.

\begin{proposition}\label{timeslice}
The algebra $\mathcal{A}^{CP}_{on}(\bH^4)$ is causal, it fulfills the time-slice axiom and it satisfies the F-locality property, namely $\mathcal{A}^{CP}_{on}(\bH^4\cap O)$ is isomorphic to $\mathcal{A}^{KG}_{on}(\bH^4\cap O)$ where $O$ is any globally hyperbolic subregion of $\bH^4$. The isomorphism is implemented by the identity.
\end{proposition}

\begin{proof}
$\mathcal{A}^{CP}_{on}(\bH^4)$ is causal, since, for any two generators $F_{[f]}, F_{[f^\prime]}$, $[f],[f^\prime]\in\frac{\mathcal{C}^{CP}_0(\bH^4)}{P[\mathcal{C}^{CP}_0(\bH^4)]}$ such that there exists two representatives $f,f^\prime\in C^\infty_{0,-}(\bH^4)$ which are spacelike separated, $F_{[f]}\star_{\bH^4}F_{[f^\prime]}-F_{[f^\prime]}\star_{\bH^4}F_{[f]} = i\langle f, E_{\bH^4} f^\prime\rangle=0$. This descends from $\supp(f)\cap\left(\supp(E(f^\prime))\cup\supp(E(\iota_z(f^\prime)))\right)=\emptyset$.
In order to prove the time-slice axiom, we need to show that, given any geodesically convex neighbourhood $\mathcal{N}$ of a Cauchy surface $\Sigma$ in Minkowski spacetime, then $\mathcal{A}^{CP}_{on}(\mathcal{N}\cap\bH^4)=\mathcal{A}_{on}^{CP}(\bH^4)$ where $\mathcal{A}^{CP}_{on}(\mathcal{N}\cap\bH^4)$ is the subalgebra of $\mathcal{A}^{CP}_{on}(\bH^4)$ obtained by considering only those $f\in\mathcal{C}^{CP}_0(\bH^4)$ such that $\supp(f)\subset\mathcal{N}$. 
In view of Corollary \ref{crucialedirei} and of Proposition \ref{obsCP}, this is equivalent to considering any $F_{[\zeta]}\in\mathcal{O}^{KG}_-(\bR^4)$ and showing that there exists at least a representative of the label $[\zeta]$ whose support is contained in $\mathcal{N}$. Let us thus fix any $\Sigma$ and $\mathcal{N}$ as above and let us consider two Cauchy surfaces $\Sigma^\pm$ such that 
$\Sigma\subset J^+(\Sigma^-)\cap J^-(\Sigma^+)\subset \mathcal{N}$. Choose $\chi^+\in C^\infty(\bR^4)$ such that $\chi^+$ is $z$-independent and $\chi^+ = 1$ for all points in $J^+(\Sigma^+)$ while it vanishes on $J^-(\Sigma^-)$. Let us consider any $[\zeta]\in\mathcal{O}^{KG}_-(\bR^4)$ and any of its representatives which we indicate with $\alpha$. Define the new function
\begin{equation}\label{aux1}
\widetilde\zeta\doteq\zeta - P \left(E^-(\zeta)+\chi^+ E(\zeta)\right),
\end{equation}
where $E^-$ is the retarded fundamental solution of $P$. Notice that, per construction and on account of the support properties of both $E^\pm$ and $\chi$, $\widetilde\zeta\in C^\infty_{0,-}(\bR^4\cap\mathcal{N})$ and it is a representative of $[\zeta]$.
We are left to prove that $\mathcal{A}^{KG}_{on}(\bR^4\cap O)$ is isomorphic to $\mathcal{A}^{CP}_{on}(\bH^4\cap O)$.
Each of these algebras is generated by those functionals whose labeling space is $C^\infty_0(O)$ and the identity map represents an isomorphism of topological vector spaces. Since the $*$-operation is complex conjugation, which is not affected by the identity map, to conclude the proof, it suffices to exhibit the following chain of identities: Let $f,f^\prime\in C^\infty_0(O)$ and let $F_f$ and $F_{f^\prime}$ be the associated generators in $\mathcal{A}^{CP}_{on}(\bH^4\cap O)$, then, for any $u\in\mathcal{C}^{CP}(\bH^4)$
\begin{equation}\label{identity}
\left(F_f\star_{\bH^4}F_{f^\prime}\right)[u]=F_f(u)F_{f^\prime}(u)+\frac{i}{2}\langle f,E_{\bH^4}(f^\prime)\rangle=\left(F_f\star F_{f^\prime}\right)[u].
\end{equation}
The last equality descends from 
$$\langle f,E_{\bH^4}(f^\prime)\rangle=\langle \eta^\dagger f, E(\eta^\dagger f^\prime)\rangle=\langle f,E(f^\prime)\rangle,$$
which holds true since $\iota_z(O)$ is causally disjoint from $O$. Notice that \eqref{identity} entails that the isomorphism between $\mathcal{A}^{CP}_{on}(\bH^4\cap O)$ and $\mathcal{A}^{KG}_{on}(\bH^4\cap O)$ is implemented by the identity map.
\end{proof}

\subsection{Hadamard states for a Casimir-Polder system}

Having constructed the algebra of observables for a Casimir-Polder system, we can focus on discussing algebraic states thereon, namely any linear functional $\omega:\mathcal{A}^{CP}(\bH^4)\to\bC$ for which
$$\omega(\mathbb{I})=1,\quad \omega(a^*a)\geq 0,\;\forall a \in\mathcal{A}^{CP}(\bH^4),$$
where $\mathbb{I}$ is the identity element. As for the usual free field theories on any globally hyperbolic spacetime, the key question is under which conditions $\omega$ is physically acceptable. We recall that the answer for $\mathcal{A}^{KG}(\bR^4)$, the algebra of observables for a Klein-Gordon field on the whole Minkowski spacetime, goes under the name of {\em Hadamard states}. More precisely, assigning a positive and normalized functional $\widetilde\omega:\mathcal{A}^{KG}(\bR^4)\to\bC$ is done via its $n$-point functions $\widetilde\omega_n:C^\infty_0(\bR^4;\bC)^{\otimes n}\to\bC$ which are chosen in such a way to encode consistently both the canonical commutation relations built in the $\star$-product. Furthermore, if all $\widetilde\omega_n$ fulfill also the equations of motion in a weak sense, $\widetilde\omega$ descends consistently to a state on $\mathcal{A}^{KG}_{on}(\bR^4)$. 

In the class of all algebraic states for $\mathcal{A}^{KG}(\bR^4)$, distinguished are the Gaussian/quasifree ones, namely whose for which the odd $n$-point functions are vanishing and the even ones can be built in terms of the $2$-point function via the following expression:
$$\widetilde\omega_{2n}(f_1\otimes...\otimes f_{2n})=\sum\limits_{\pi_{2n}\in S^\prime_{2n}}\prod\limits_{i=1}^n\widetilde\omega_2\left(f_{\pi_{2n}(i-1)}\otimes f_{\pi_{2n}(i)}\right),$$
where $S^\prime_{2n}$ stands for the set of ordered permutations of $2n$-elements. In between all quasi-free states, those of Hadamard form can be characterized out of the singular structure of the bi-distribution $\widetilde\omega_2\in\mathcal{D}^\prime(\bR^4\times\bR^4)$ associated to the two-point function $\widetilde\omega_2$ via the Schwarz kernel theorem \cite{Radzikowski:1996pa, Radzikowski:1996ei}, that is 
\begin{equation}\label{WF}
WF(\widetilde\omega_2)=\left\{(x,x^\prime,k_x,-k_{x^\prime})\in T^*(\bR^4\times\bR^4)\setminus \{\mathbf 0\}\;|\;
(x,k_x)\sim(x^\prime,k_{x^\prime}),\;k_x\triangleright 0\right\},
\end{equation}
where $\sim$ entails that $x$ and $x^\prime$ are connected via lightlike geodesic and $\eta^{-1}(k_{x^\prime})$ is the parallel transport of $\eta^{-1}(k_x)$ along it. The symbol $\triangleright$ entails that $k_x$ is a future pointing covector. Notice that \eqref{WF} can be straightforwardly extended to any bi-distributions defined on any globally hyperbolic spacetime and that, if we add the requirement, that $\widetilde\omega_2$ is a weak bi-solution of the equation of motion ruled by $P$, then $\widetilde\omega$ identifies a state also for $\mathcal{A}^{KG}_{on}(\bR^4)$

\vspace{0.2cm}

In view of the structure of $\bH^4$, extending the above considerations to $\mathcal{A}^{CP}(\bH^4)$ is not straightforward. A similar problem appeared in Abelian gauge theories \cite{Fewster:2003ey} or in linearized gravity \cite{Benini:2014rya, Fewster:2012bj}. The way out that we propose is partly inspired by these papers, partly by F-locality: We require that a physically acceptable, quasi-free state on $\mathcal{A}^{CP}(\bH^4)$ is such that its restriction to any globally hyperbolic subregion of $\bH^4$ descends from a bi-distribution, thereon of Hadamard form. 

\begin{definition}\label{HadCP}
We call a linear map $\omega:\mathcal{A}^{CP}(\bH^4)\to\bC$ a quasi-free {\bf Hadamard state for a Casimir-Polder system} if it is normalized, positive, quasi-free and if, for all globally hyperbolic submanifolds $O\subset\bH^4$, the restriction of $\omega$ to $\mathcal{A}^{CP}(\bH^4\cap O)$ is such that there exists $\omega_2\in\mathcal{D}^\prime(O\times O)$ whose wavefront set is of Hadamard form
$$WF(\omega_2)=\left\{(x,x^\prime,k_x,-k_{x^\prime})\in T^*(O\times O)\setminus \{\mathbf 0\}\;|\;
(x,k_x)\sim(x^\prime,k_{x^\prime}),\;k_x\triangleright 0\right\}, $$
 and, for all $F_{f},F_{f^\prime}\in\mathcal{A}^{CP}(O)$
$$\omega\left(F_f\star_{\bH^4} F_{f^\prime}\right)=\omega_2(f,f^\prime).$$
\end{definition}

Notice that, in order for $\omega$ to descend to a state on $\mathcal{A}^{CP}_{on}(\bH^4)$ a compatibility condition with the equations of motion must be required. In view of this last definition the first question to answer is whether one can build a connection between Hadamard states for the on-shell algebra of the Klein-Gordon field on Minkowski spacetime and that of a Casimir-Polder system. 

\begin{proposition}\label{comparison}
Let $\widetilde\eta^*:\mathcal{A}^{CP}_{on}(\bH^4)\to\mathcal{A}^{KG}_{on}(\bR^4)$ be the map defined in Proposition \ref{globalcomparison}. Then, for every quasi-free state $\widetilde\omega:\mathcal{A}^{KG}_{on}(\bR^4)\to\bC$, there exists a quasi-free state $\omega$ on $\mathcal{A}^{CP}_{on}(\bH^4)$ such that for all $a\in\mathcal{A}^{CP}_{on}(\bH^4)$, $\omega(a)\doteq \widetilde\omega(\widetilde\eta^*(a))$. In particular, if $\widetilde\omega$ is of Hadamard form, so is $\omega$.
\end{proposition}

\begin{proof}
As a starting point, notice that $\omega$ inherits the normalization, positivity and the property of being quasi-free directly from $\widetilde\omega$. We need only to check the Hadamard property. Let $O\subset\bH^4$ be any globally hyperbolic submanifold; for every $F_f,F_{f^\prime}\in\mathcal{A}^{CP}(O)$ 
$$\omega(F_f\star_{\bH^4} F_{f^\prime})=\widetilde\omega(\eta^*(F_f\star_{\bH^4} F_{f^\prime}))=\widetilde\omega(F_f\star F_{f^\prime})=\frac{1}{2}\widetilde\omega_2(f-\iota_z(f),f^\prime-\iota_{z}(f^\prime)).$$
In other words the bi-distribution associated to $\omega$ can be built out of $\widetilde\omega_2$ itself. Since the latter has per hypothesis the Hadamard wavefront set and since, if $\supp(f),\supp(f^\prime)\subset O\subset\bH^4$, then neither $\iota_z(f)$ nor $\iota_z(f^\prime)$ can be entirely supported therein, the only singular term in the above identity is $\widetilde\omega_2(f,f^\prime)$. Hence $\omega$ is of Hadamard form. 
\end{proof}

As a last step, we wish to compare our approach with the \textbf{\em method of images} which is commonly used on Minkowski spacetime. 

\begin{lemma}\label{comparison-Mink}
Let $\widetilde\omega$ be any quasi-free Hadamard state for $\mathcal{A}^{KG}(\bR^4)$ whose associated two-point function  $\widetilde\omega_2\in\mathcal{D}^\prime(\bR^4\times\bR^4)$ has an integral kernel which is invariant under reflection in both entries along the $z$-direction, that is $\widetilde\omega_2(\underline{x},z,\underline{x}^\prime,z^\prime)=\widetilde\omega_2(\underline{x},-z,\underline{x}^\prime,-z^\prime)$. Then the state $\omega$ on $\mathcal{A}^{CP}(\bH^4)$ built as per Proposition \ref{comparison} is a quasi-free Hadamard state whose integral kernel is
\begin{equation}\label{int_ker}
\omega_2(\underline{x},z,\underline{x}^\prime,z^\prime)=\widetilde\omega_2(\underline{x},z,\underline{x}^\prime,z^\prime)-\widetilde\omega_2(\underline{x},-z,\underline{x}^\prime,z^\prime).
\end{equation}
\end{lemma}

\begin{proof}
On account of Proposition \ref{comparison}, we can conclude that $\omega$ is a Hadamard state on $\mathcal{A}^{CP}(\bH^4)$ and it is quasi-free per construction. In order to show the last statement, it suffices instead an explicit calculation. Let $\omega$ be as per hypothesis and let $\omega_2$ be the associated bi-distribution. For all $f,f^\prime\in\eta[C^\infty_{0}(\bR^4)]$, seen as labels for two generators of $\mathcal{A}^{CP}(\bH^4)$, it holds in view of Proposition \ref{comparison}
\begin{gather*}
\omega(F_f\star_{\bH^4} F_{f^\prime})=\widetilde\omega(\eta^*\left(F_f\star_{\bH^4} F_{f^\prime}\right))=
\frac{1}{2}\widetilde\omega_2(f-\iota_z(f),f^\prime-\iota_z(f^\prime))=\widetilde\omega_2(f-\iota_z(f),f^\prime),
\end{gather*}
where, in the last equality, we used the symmetry hypothesis of the two-point function to conclude that $\omega_2(f,f^\prime)=\omega_2(\iota_z(f),\iota_z(f^\prime))$ and $\omega_2(f,\iota_z(f^\prime))=\omega_2(\iota_z(f),f^\prime)$. The above chain of equalities entails the sought identity at a level of integral kernels.
\end{proof}

\begin{remark}
The statement of Lemma \ref{comparison-Mink} applies to the Poincar\'e vacuum and the KMS state for a massive or massless Klein-Gordon field on Minkowski spacetime, for which $\widetilde\omega_2$ induces the same quasi-free state which one obtains via the method of images.
\end{remark}

For completeness, we want now to discuss the form of the singular structure of the two-point function of the states obtained in Lemma \ref{comparison-Mink}. In view of \eqref{int_ker} we have that 
\begin{equation} \label{eq:WF-reflection}
WF(\omega_2) =  WF(\left.\widetilde\omega_2 \right|_{\mathbb{H}^4})  \cup WF(\left(\widetilde\omega_2 \circ\left(\iota_z\otimes\mathbb{I}\right)\right)\left.\right|_{\mathbb{H}^4})
\end{equation}
where the restriction map refers to the points of the singular support. Furthermore
\begin{gather*}
 WF(\left.\widetilde\omega_2 \circ\left(\iota_z\otimes\mathbb{I}\right) \right|_{\mathbb{H}^4}) =\\ \left\{ (x,x',k_x,-k'_{x^\prime}) \in T^*\left({\mathbb{H}^4\times\mathbb{H}^4}\right)\setminus\{\mathbf 0\}\;|\; (x,k_x)\sim (\iota_z x^\prime,(\iota_z)_\ast k'_{x^\prime}) , k_x \triangleright 0    \right\} 
=  (\iota_z\otimes\mathbb{I}) WF(\left.\widetilde\omega_2 \right|_{\mathbb{H}^4}).
\end{gather*}
In the previous expression, $\iota_z$ acts on covectors inverting the sign of the $z-$component.
Heuristically, we might say that $(x,x';k_x,k'_{x^\prime})$ are contained in $WF(\left(\widetilde\omega_2\left(\iota_z\otimes\mathbb{I}\right)\right)\left. \right|_{\mathbb{H}^4})$ if and only if $x$ and $x'$ are connected by a null geodesic {\em reflected} at the surface $\partial \mathbb{H}$ and if $\eta^{-1}(k_x)$ and $\eta^{-1}(-k'_{x^\prime})$ are tangent vectors at the end points of this reflected geodesic.
Notice that, whenever $\omega_2$ is restricted to a globally hyperbolic region $O\subset\mathbb{H}$, its wave front set enjoys the microlocal spectrum condition because $WF(\left(\widetilde\omega_2 \circ\left(\iota_z\otimes\mathbb{I}\right)\right)\left. \right|_{O})$ is the empty set. No lightlike geodesic starting from $O$ can re-enter after reflection.

\subsection{Wick ordering in a Casimir-Polder system}\label{Wick1} 

To conclude the section, we show how to make contact between the previous analysis and the standard results in the literature concerning the Casimir-Polder energy. To this avail, we need first of all to introduce the (local) Wick polynomials for a Casimir-Polder system. 
From a conceptual point of view, this question is the same as for a Klein-Gordon field on a globally hyperbolic spacetime. 
We shall see, however, that, on every globally hyperbolic submanifolds of $\mathbb{H}^4$, the local Wick monomials generate an algebra of observables which is isomorphic to the restriction thereon of the Klein-Gordon one. Hence it is well-defined. Yet, in order to build a global algebra of Wick polynomials, one has to take into account that, on account of the presence of the boundary conditions, it is not possible to define a global Hadamard function which depends only on local properties of the spacetime. We shall show that such obstacle can be circumvented, though at the price that the the embedding of the local algebras into the global one involves a non-local deformation. 
  
Before entering into the technical details, also to make contact with the standard literature on the Casimir-Polder effects, it is worth recalling a few facts valid for any scalar field theory on a globally hyperbolic spacetime $(M,g)$:
\begin{itemize}
\item Let $\widetilde\omega_2\in\mathcal{D}^\prime(M\times M)$ be a bi-distribution which induces a quasi-free state $\widetilde\omega$ on the $*$-algebra of fields $\mathcal{A}^{KG}(M)$, which is defined in full analogy with the one introduced on Minkowski spacetime. If the wavefront set of $\widetilde\omega_2$ is of the form \eqref{WF} and thus $\widetilde\omega$ is Hadamard, it is possible to give a rather explicit local characterization to the integral kernel of $\widetilde\omega_2$.

For every pair of points $x,y\in M$ lying in the same geodesic neighbourhood, we can split $\widetilde\omega_2(x,x^\prime)$ as follows -- see for example \cite{Moretti:2001qh}: $\widetilde\omega_2(x,x^\prime)=H(x,x^\prime)+W(x,x^\prime)$, where $W(x,x^\prime)$ is a smooth term while $H(x,x^\prime)$ is a (Hadamard) parametrix. It is a singular term which depends only on the background geometry and on the partial differential operator $P$ ruling the dynamics,
\item Let $\mathcal{U}\subset M$ be a geodesic convex, open neighbourhood. Following \cite{Hollands:2001nf}, we define the normal ordered squared Wick polynomial via the map 
$$
f \in C^\infty_0(\mathcal{U})\mapsto\normal{\widehat{\phi}^2}_H(f)\doteq\int\limits_{M\times M}d\mu_g(x)\,d\mu_g(x^\prime)\,\left(\widehat\phi(x)\widehat\phi(x^\prime)-H(x,x^\prime)\right)f(x)\delta(x,x^\prime),
$$
where $d\mu_g$ is the metric induced measure and the integral is taken over the whole manifold on account of the support properties of $f$. 
Notice that the expectation value  
\begin{equation}\label{eq:n-order}
\widetilde\omega(\normal{\widehat{\phi}^2}_H(f))
=
\int\limits_{M\times M}d\mu_g(x)\,d\mu_g(x^\prime)\,\left(\widetilde\omega(x,x^\prime)-H(x,x^\prime)\right)f(x)\delta(x,x^\prime)
\end{equation}
is well defined when computed on Hadamard states. Furthermore, we recall that the Hadamard parametrix is uniquely determined up to smooth terms yielding the standard regularization freedoms \cite{Hollands:2001nf}. In particular, if $(M,g)$ is Minkowski spacetime and $\widetilde\omega^0$ is the Poincar\'e vacuum for a massless Klein-Gordon field, then $H(x,x^\prime)$ can be chosen to be equal to  $\widetilde\omega_2(x,x^\prime)$ and in this case, for all $f \in C^\infty_0(\bR^4)$, $\normal{\widehat{\phi}^2}_{H}(f)$ vanishes up to the regularization freedom. The goal of the functional approach is to recollect all observables which are regularized in a coherent body, endowing it with the structure of a $*$-algebra. In other words it encodes the so-called Wick theorem.
\end{itemize}

An elegant way of introducing an algebraic structure on the set of Wick polynomials is provided by methods of perturbative algebraic quantum field theory. We recall here this construction for the Klein-Gordon case on the whole Minkowski spacetime. Notice that, in order to encompass Wick polynomials in the algebra of functionals, it would be desirable to extend $\mathcal{A}^{KG}(\bR^4)$ adding non linear local generators like 
\begin{equation}\label{phi2}
F^{(2)}_f(\phi) = \int_{\mathbb{H}^4} d\mu_g(x) \phi^2(x) f(x) ,  
\end{equation}
where $f\in C^\infty_0(\bR^4)$ while $\phi\in\mathcal{S}(\bR^4)$. The composition of two of these functionals via the $\star-$product introduced in \eqref{algprod2} is, however, ill-defined at a microlocal level. In order to overcome this difficulty, we follow in \cite{ Brunetti:2009qc, Brunetti:2009pn, Fredenhagen:2012sb}, modifying the composition rule in $\mathcal{A}^{KG}(\bR^4)$ and then enlarging such set to include also additional regularized fields.
The sought modification must preserve the commutation relations among the generators of $\mathcal{A}^{KG}(\bR^4)$. It can be written as in \eqref{algprod} with $\Gamma_E$ replaced by 
$$
\Gamma_{H}=-i \int\limits_{\bR^4\times\bR^4} H(x,x^\prime)\frac{\delta}{\delta \phi(x)}\otimes\frac{\delta}{\delta \phi(x^\prime)}.
$$
The product obtained in this way is denoted by $\star_H$ and on $\mathcal{A}^{KG}(\bR^4)$ takes the same form given in \eqref{algprod2} where the integral kernel $E(x,x^\prime)$ is replaced by $-2i H(x,x^\prime)$,  up to multiplicative constants the (global) Hadamard parametrix. Notice that the antisymmetric part of $-2iH(x,x^\prime)$ coincides with $E(x,x^\prime)$ and hence the canonical commutation relations among the generators of $\mathcal{A}^{KG}(\bR^4)$ are left untouched. Furthermore, since the new $\star$-product is built only out of local structures, covariance of the scheme is guaranteed. In addition, the form \eqref{WF} of the wavefront set of $H(x,x^\prime)$ entails that powers of $H(x,x^\prime)$ are meaningful since the H\"ormander criterion for multiplication of distributions is satisfied -- see \cite[Th. 8.2.10]{Hormander1}.

Equipping $\mathcal{F}_0(\bR^4)$ with the product $\star_H$ instead of 
the original $\star$ we obtain an algebra which is isomorphic to $\mathcal{A}^{KG}(\bR^4)$. Furthermore, following \cite{Brunetti:2009pn}, this isomorphism can be understood as a deformation of the original algebra $\mathcal{A}^{KG}(\bR^4)$ which is generated by 
\begin{equation}\label{eq:def-alpha}
\alpha_{H}\doteq \sum\limits_{n=0}^\infty \frac{\Gamma_{H}^n}{ n!}:\mathcal{A}^{KG}\to\mathcal{A}^{KG}
\end{equation}
via
\begin{equation}\label{deformed}
\left(F\star_H F^\prime\right)=\alpha_H\left(\alpha_H^{-1}(F)\star\alpha_H^{-1}(F^\prime)\right).
\end{equation}
After such deformation, the set of elements constituting the algebra can be enriched by adding also local non linear functionals like those of the form \eqref{phi2}.  For completeness, we recall the form of the set on which, after the deformation, the algebra of fields can be extended. 
\begin{definition}
We call {\bf microcausal functionals} for the Klein-Gordon field, $\mathcal{A}_\mu^{KG}(\bR^4)$, the collection of all smooth functionals $F:\mathcal{C}^{KG}(\bR^4)\to\bC$ such for all $n\geq 1$ and for all $\phi\in\mathcal{C}^{KG}(\bR^4)$, $F^{(n)}[\phi]\in\mathcal{E}^\prime(\bR^4)^{\otimes n}$. Only a finite number of functional derivatives do not vanish and $\textrm{WF}(F^{(n)})\subset\Xi_n$,  where
$$
\Xi_n\doteq T^*(\bR^4)^n\setminus\left\{(x_1,...,x_n,k_1,...,k_n)\,|\;(k_1,...,k_n)\in\left.\left(\overline{V}^n_+\cup\overline{V}^n_-\right)\right|_{(x_1,...,x_n)}\right\},
$$ 
where $\overline{V}_\pm$ are the subsets of $T^*\bR^4$ formed by elements $(x_i,k_i)$ where each covector $k_i$, $i=1,...,n$ lies in the closed future ($+$) and in the closed past ($-$) light cone. The pair $(\mathcal{A}_\mu^{KG}(\bR^4),\star_H)$ is called {\bf extended algebra of Wick polynomials}.
\end{definition}

Notice that the expectation values of products of generators of $\mathcal{A}^{KG}(\bR^4)$ with respect to a state $\omega$ must be invariant under the deformation. In other words $\mathcal{A}^{KG}_\mu(\bR^4)$ contains a $*$-subalgebra isomorphic to $\mathcal{A}^{KG}(\bR^4)$. As a last remark on this procedure we stress that, since only the antisymmetric part of $H(x,x^\prime)$ is fixed, there is a freedom in the definition of the extended objects. This is related to the renown {\bf regularization freedom}, a discussion of which can be found for example in \cite{Hollands:2001nf}. In this paper we will not enter into the details, since they are not necessary to our purposes. 

We recall that the procedure discussed so far can be applied almost slavishly on every globally hyperbolic spacetime. Hence, as far as a Casimir-Polder system is concerned, our strategy is to start by constructing an extended algebra of Wick polynomials starting from any $*$-algebra $\mathcal{A}^{CP}(O)$ as in Proposition \ref{timeslice}. Recall that $O$ is a globally hyperbolic submanifold of $\bH^4$. 

To this avail we recall the definition of support for functionals as introduced in \cite{Fredenhagen:2012sb} and adapted to our case.
\begin{definition}\label{support}
Let $F:\mathcal{C}^{CP}(\bH^4)\to\bC$ be any functional on the space of off-shell configurations for a Casimir-Polder system as per Definition \ref{CPoff}. We call {\em support} of $F$
\begin{gather*}
\supp(F)\doteq\{x\in\bH^4\;|\;\forall\,\textrm{neighbourhoods}\, U\ni x,\,\exists u,u^\prime\in\mathcal{C}^{CP}(\bH^4),\,\supp(u)\subseteq U,\\
\;\textrm{such that}\;F[u+u^\prime ]\neq F[u]\}.
\end{gather*}
\end{definition}

Let $O\subset\bH^4$ be any globally hyperbolic submanifold, to which we associate $\mathcal{A}^{CP}(O)$ as per Proposition \ref{timeslice}. In view of Definition \ref{HadCP} we follow the same procedure, used to build for $\mathcal{A}_\mu^{KG}(\bR^4)$, to obtain $\mathcal{A}_\mu^{CP}(O)$ an extended algebra of Wick polynomials. Furthermore, in view of Proposition \ref{timeslice},  $\mathcal{A}_\mu^{CP}(O)$ is $*$-isomorphic to $\mathcal{A}_\mu^{KG}(O)$, the restriction of $\mathcal{A}_\mu^{KG}(\bR^4)$ to $O$. 
 
The next step consists of gluing together all $\mathcal{A}^{CP}_\mu(O)$, so to obtain a global extend algebra of Wick polynomials for a Casimir-Polder system. 
The following remark shows that an obstruction arises in considering $\star_H$ 
as the product for the global extended algebra. It turns out that the gluing becomes possible only after a suitable deformation of $\star_H$.

\begin{remark}
Let $O_1$ and $O_2$ be two globally hyperbolic submanifolds of $\mathbb{H}^4$ whose union is not contained in a third globally hyperbolic submanifold of $\mathbb{H}^4$. Consider now $F_f^{(2)}\in \mathcal{A}^{CP}_\mu(O_1,\star_H)$  and $F_{f^\prime}^{(2)}\in \mathcal{A}^{CP}_\mu(O_2,\star_H)$ such that
\[
F_f^{(2)}(u)\doteq  \int_{\mathbb{H}^4} d\mu_g(x) f(x) u^2(x) ,\qquad F_{f^\prime}^{(2)}(u)\doteq  \int_{\mathbb{H}^4} f^\prime(x) d\mu_g(x) u^2(x) ,
\] 
where $u\in\mathcal{C}^{CP}(\bH^4)$ while $\supp{f}\subset O_1$ and $\supp{f^\prime}\subset O_2$. In view of Proposition \ref{timeslice}, we choose the Hadamard parametrix $H(x,x^\prime)$ to be the same one as for a Klein-Gordon scalar field on Minkowkski spacetime, though restricted to the region(s) of interest. In order to compute the correlations between the above two elements, we need to recognize them as being part of a larger extended algebra.
Yet, if we try to follow the same procedure used in \eqref{deformed} for $\mathcal{A}^{KG}_\mu(\bR^4)$, we notice that the local $\star$-products for the undeformed algebra are defined replacing $\frac{i}{2}E_{\mathbb{H}^4}(x,x^\prime)$ with
\[
H(x,x^\prime) + \frac{i}{2}\left(E_{\mathbb{H}^4}(x,x^\prime) -E(x,x^\prime) \right).
\]
Pathologies in the computation of $F_f^{(2)}\star_H F_{f'}^{(2)}$ occur, as terms including $\left(E_{\mathbb{H}^4}(x,x^\prime) -E(x,x^\prime) \right)$ multiplied with itself do appear. They are ill-defined.

Such obstructions\footnote{We are grateful to Chris Fewster for suggesting this approach.} can be removed exploiting the fact that algebras whose $\star-$products are constructed with different Hadamard functions are $*-$isomorphic \cite{Brunetti:2009qc}. Mimicking the construction of $E_{\mathbb{H}^4}$ starting from $E$, and in view of  \eqref{int_ker}, let us consider the bidistribution $H_{\mathbb{H}^4}$ whose integral kernel is
\[
H_{\mathbb{H}^4}(\underline{x},z,\underline{x}^\prime,z^\prime)
\doteq
H(\underline{x},z,\underline{x}^\prime,z^\prime)-
H(\underline{x},-z,\underline{x}^\prime,z^\prime).
\]
Notice that $H_{\mathbb{H}^4}$ yields $\mathcal{A}^{CP}_\mu(\mathbb{H}^4,\star_{H_{\mathbb{H}^4}})$, a well defined global algebra. Hence, the correlations among elements of $\mathcal{A}^{CP}_\mu(O_1,\star_H)$ and of $\mathcal{A}^{CP}_\mu(O_2,\star_H)$ are meaningful only if we embed them in $\mathcal{A}^{CP}_\mu(\mathbb{H}^4,\star_{H_{\mathbb{H}^4}})$. Such embedding is realized by $\alpha_{H_{\mathbb{H}^4}-H}$ as in \eqref{eq:def-alpha} and it is an injective $*-$isomorphism.  
\end{remark}

Despite this hurdle, concepts like smeared energy density are still well-defined within each $\mathcal{A}^{CP}_\mu(O)$. 
Furthermore, regardless of the existence of an extended algebra of observables, well-known blows-up in computing quantities, such as the Casimir total energy, still remain due to additional divergences present in observables supported on the boundaries. 

\bigskip

We can make finally a correspondence to the standard results in the literature, in particular recovering the dependence of the energy density on the forth power of the distance along the $z$-axis between a point in the bulk and one on the boundary. Before stating the result, we recall that, on Minkowski spacetime, the so-called {\em improved stress-energy tensor} of a massless conformally coupled scalar field $\phi$ is on-shell \cite{Callan:1970ze, Moretti:2001qh}
\begin{equation}\label{improved}
  T_{\mu\nu}=\partial_\mu \phi \partial_\nu \phi - \frac12 \eta_{\mu\nu} \partial^\rho \phi \partial_\rho \phi %
        + \xi (\eta_{\mu\nu} \Box - \partial_\mu \partial_\nu )\phi^2,
\end{equation}
where $\xi$ is the coupling constant with the scalar curvature $R$ introduced in \eqref{dynCP}.

\begin{lemma}\label{reale}
Let us consider a massless, arbitrarily coupled to scalar curvature, scalar field and let $\omega^0$ be the Hadamard state for $\mathcal{A}^{CP}(\bH^4)$ induced from the Poincar\'e vacuum $\widetilde\omega^0$ via Lemma \ref{comparison-Mink}. Let $\widetilde\omega^0_{2}(x,x^\prime)$ be the associated integral kernel of the two-point function on the whole $\bR^4$. Then, for all $f\in C^\infty_{0}(\mathring{\bH}^4)$,
$$
\omega^0(\normal{\widehat{\phi}^2}_{H}(f))=-\frac{1}{32\pi^2}\int\limits_{\bR^4}d^4x \frac{f(\underline{x},z)}{z^2},
$$
and
$$\omega^0(\normal{\widehat{T_{\mu\nu}}}_{H}(f))=A_{\mu\nu}\frac{6\xi-1}{32\pi^2}\int\limits_{\bR^4}d^4x \frac{f(\underline{x},z)}{z^4},$$
where $\{T_{\mu\nu}\}$ are the components of the stress-energy tensor \eqref{improved} while $A=\mathrm{diag}(-1,1,1,0)$.
\end{lemma}

\begin{proof}
We need only to recollect what already proven together with the explicit form of
\begin{equation}\label{vacuum}
\widetilde\omega^0_{2}(x,x^\prime)=\lim\limits_{\epsilon\to 0^+}\frac{1}{4\pi^2}\frac{1}{\left(\eta_3^{\mu\nu}(\underline{x}_\mu-\underline{x}^\prime_\nu)+(z-z^\prime)^2\right)+i\epsilon(\underline{x}_0-\underline{x}^\prime_0)+\epsilon^2},
\end{equation}
where $\eta_3=\textrm{diag}(-1,1,1)$. On account of both Proposition \ref{comparison} and Lemma \ref{comparison-Mink}, we know that $\omega^0$ is a Hadamard state for $\mathcal{A}^{CP}(\bH^4)$. The definition of $\mathcal{A}^{CP}_\mu(\bH^4)$ together with both $\widetilde\omega^0_{2}(x,x^\prime)=H(\underline{x},z,\underline{x}^\prime,z^\prime)=H(\underline{x},-z,\underline{x}^\prime,-z^\prime)$ and Proposition \ref{comparison} entail that, calling $\zeta=f-\iota_z(f)\in C^\infty_0(\bR^4)$,
\begin{gather*}
\omega^0(\normal{\widehat{\phi}^2}_H(\zeta))\doteq
\int\limits_{\bR^4\times\bR^4}d^4x\,d^4x^\prime\,\left(\widetilde\omega^0_{2}(x,x^\prime)-\frac{1}{2}H(\underline{x},z,\underline{x}^\prime,z^\prime)\right)\zeta(x)\delta(x-x^\prime)=\\
=-\int\limits_{\bR^4\times\bR^4}d^4x\,d^4x^\prime\,\left(\frac{1}{4}\left(\omega^0_{2}(\underline{x},-z,\underline{x}^\prime,z^\prime)+\omega^0_{2}(\underline{x},z,\underline{x}^\prime,-z^\prime)\right)\right)\zeta(x)\delta(x-x^\prime)=\\
=-\int\limits_{\bR^4\times\bR^4}d^4x\,d^4x^\prime\,\frac{1}{2}H(\underline{x},-z,\underline{x}^\prime,z^\prime)\zeta(x)\delta(x-x^\prime)=-\frac{1}{32\pi^2}\int\limits_{\bR^4}d^4x \frac{\zeta(\underline{x},z)}{z^2}
\end{gather*}
In order to compute $\omega^0(\normal{\widehat{T_{\mu\nu}}}_{H}(\zeta))$ it suffices to apply the point-splitting scheme as introduced in \cite{Moretti:2001qh}. All results obtained in this cited paper apply without modifications to the case at hand. In particular it holds that
$$\omega^0(\normal{\widehat{T_{\mu\nu}}}_{H}(\zeta))=\int\limits_{\bR^4\times\bR^4}d^4x\,d^4x^\prime\,\left(D_{\mu\nu}^{(x,x^\prime)}\left(\widetilde\omega^0_{2}(x,x^\prime)-\frac{1}{2}H(\underline{x},z,\underline{x}^\prime,z^\prime)\right)\right)\zeta(x)\delta(x-x^\prime),$$
where -- see \cite[\S 4]{Hack:2010iw}
\begin{equation}\label{D00}
D_{\mu\nu}^{(x,x^\prime)}=\frac{\partial}{\partial x^\mu}\frac{\partial}{\partial x'^\nu}-\frac12 \eta_{\mu\nu}\eta^{\alpha\beta}\frac{\partial}{\partial x^\alpha}\frac{\partial}{\partial x'^\beta}+\xi \left( \eta_{\mu\nu}\eta^{\rho\lambda}\frac{\partial}{\partial x^\rho}\frac{\partial}{\partial x^\lambda} - \frac{\partial}{\partial x^\mu}\frac{\partial}{\partial x^\nu}\right).
\end{equation}
Inserting this expression in the above integral and replacing $\widetilde\omega^0_{2}(x,x^\prime)-\frac{1}{2}H(\underline{x},z,\underline{x}^\prime,z^\prime)$ with $\frac{1}{2}H(\underline{x},-z,\underline{x}^\prime,z^\prime)$ yields the sought result.
\end{proof}

\begin{remark}\label{alternative}
Notice that we have defined the Wick polynomials only for those smooth and compactly supported functions whose support does not intersect the boundary of the region of interest. The reason can be seen explicitly looking at the last lemma: If we inspect the integral kernels $\widetilde\omega_2(\underline{x},-z,\underline{x}^\prime,z^\prime)$ and $\widetilde\omega_2(\underline{x},z,\underline{x}^\prime,-z^\prime)$, they become singular at $z=z^\prime=0$ so that they cannot be tested with $\delta(z-z^\prime)$. This is no surprise and it is at the heart of the often mentioned problem that, in a Casimir or in a Casimir-Polder system, the total energy, computed out of the integral of the time-component of the stress-energy tensor diverges.
\end{remark}


\se{Algebraic Quantum Field Theory and the Casimir effect}

In this section we shall focus on the second scenario, we are interested in, namely the one describing the attraction force between two parallel, perfectly conducting, plates as discussed for the first time in \cite{Casimir:1948dh}. We shall refer to it as {\em Casimir system}. As in the previous section we shall investigate this model from the point of view of algebraic quantum field theory and using the so-called functional formalism. Following the same path as in a Casimir-Polder system, we shall proceed in three main steps:

\vskip .2cm

\noindent\textbf{\em Part 1 -- Dynamical configurations:} At a geometric level, the model consists of the region $Z\doteq\bR^3\times [0,d]\subset\bR^4$ endowed with the (restriction of the) Minkowski metric. In analogy to the previous section, the interval $[0,d]$ runs along the spacelike $z$-direction. At a field theoretical level, our starting point are are all $u\in C^\infty(Z)$, where smoothness is meant as in Definition \ref{smoothonbound} since $Z\subset\bH^4$. {\em Dynamical configurations} are instead the elements of the vector space $\mathcal{S}^C(Z)$ built out of the smooth solutions of
\begin{equation}\label{Casimir-eom}
\left\{\begin{array}{l}
Pu=(\Box-\xi R -m^2)u=0,\quad m^2\geq 0,\quad\xi\in\bR\\
u(\underline{x},0)=u(\underline{x},d)=0
\end{array}\right. ,
\end{equation}
where $R$ is the scalar curvature. Since the scalar curvature vanishes, the term $\xi R$ plays no role at a dynamical level, but it affects the structure of the stress-energy tensor which we will consider later. 

Notice that, in full analogy with the previous section, neither $(Z,\eta)$ is  a globally hyperbolic spacetime, nor \eqref{Casimir-eom} is an initial value problem, rather it is a boundary value problem. Hence, in order to characterize $\mS^C(Z)$, we follow the same strategy used in a Casimir-Polder system, namely we identify each smooth solution of \eqref{Casimir-eom} with a specific counterpart for a Klein-Gordon field on the whole Minkowski spacetime. Before outlining the details, we introduce the auxiliary regions
\begin{equation}\label{Y-regions}
Y_0\doteq \bR^3\times [-d,d],\quad Y_n\doteq\{x\in\bR^4\;|\;\exists (\underline{x},z)\in Y_0\;\textrm{for which}\;x=(\underline{x},z+2nd)\},\;n\in\mathbb{Z}.
\end{equation}
As a consequence $\bR^4=\bigcup\limits_{n\in\bZ}Y_n$. 

\begin{proposition}\label{dynCas}
There exists a vector space isomorphism between $\mS^C(Z)$ and the quotient $\frac{C^\infty_{tc,C}(\bR^4)}{P[C^\infty_{tc,C}(\bR^4)]}$ where $C^\infty_{tc,C}(\bR^4)$ is the collection of all $\alpha\in C^\infty_{tc}(\bR^4)$ such that the following conditions are met:
\begin{enumerate}
\item $\alpha\in C^\infty_{tc,-}(\bR^4)$, that is $\alpha(\underline{x},z)=-\alpha(\underline{x},-z)$
\item $\alpha(\underline{x},z)=-\alpha(\underline{x},2d-z)$
\end{enumerate}
\end{proposition}

\begin{proof}
As a first step we show that there exists an isomorphism between $\mS^C(Z)$ and a vector subspace of $\mS^{KG}(\bR^4)\doteq\{\phi\in C^\infty(\bR^4)\;|\; P\phi = 0\}$. Let $u\in\mS^C(Z)$ and let  
$$v(x)\doteq \left\{\begin{array}{ll}
u(x), & x\in Z\\
-u(-x), & x \in Y_0\setminus Z
\end{array}\right. .$$
Following the same argument as in the proof of Proposition \ref{CPdynchara}, we can conclude that $v\in C^\infty(Y_0)$ and $v(\underline{x},0)=v(\underline{x},d)=v(\underline{x},-d)=0$. Define $\phi(x)=\phi(\underline{x},z)\doteq v(\underline{x},z-2nd),$ for any $x\in Y_n$. By a similar argument as for $v(x)$, it descends that $\phi\in C^\infty(\bR^4)$ and that, moreover, $P\phi=0$, as this property is traded from that of $u$. In other words we have found a linear map 
\begin{gather}
F:\mS^C(Z)\to \mS^C(\bR^4)\subset\mS^{KG}(\bR^4)\notag\\
\mS^C(\bR^4)=\left\{\phi\in C^\infty_-(\bR^4)\;|\;P\phi=0\;\textrm{and}\;\phi(\underline{x},2d-z)=-\phi(\underline{x},z)\right\}.\label{image-sol}
\end{gather}
The map is per construction surjective, since for every $\phi\in\mS^C(\bR^4)$, $\left.\phi\right|_Z\in\mS^C(Z)$ and $F(\left.\phi\right|_Z)=\phi$. Furthermore $F$ is also injective since $F(u)=0\in\mS^C(\bR^4)$ implies $\phi=0$ and, thus $u=\left.\phi\right|_Z=0$. In other words $F$ is an isomorphism of vector spaces. To prove the statement of the proposition we need to show that $\mS^C(\bR^4)$ is isomorphic to $\frac{C^\infty_{tc,C}(\bR^4)}{P[C^\infty_{tc,C}(\bR^4)]}$. 
As a first step we show that the map induced by $E$ is surjective. Let thus $\phi\in\mS^C(\bR^4)$. Since $P\phi=0$, there must exist $\alpha\in C^\infty_{tc}(\bR^4)$ such that $\phi=E(\alpha)$. Since $\phi$ is odd per reflection along the hyperplane $z=0$, we know from Proposition \ref{CPdynchara} that $\alpha$ must lie in $C^\infty_{tc,-}(\bR^4)$. Repeating slavishly the proof of Proposition \ref{CPdynchara} with respect to the condition $\phi(\underline{x},2d-z)=-\phi(\underline{x},z)$ we obtain that $\alpha\in C^\infty_{tc,-,d}(\bR^4)$ where $C^\infty_{tc,-,d}(\bR^4)=\left\{\alpha\in C^\infty_{tc}\;|\;\alpha(\underline{x},2d-z)=-\alpha(\underline{x},z)\right\}$. Putting all together $\alpha\in C^\infty_{tc,-}(\bR^4)\cap C^\infty_{tc,-,d}(\bR^4)=C^\infty_{tc,C}(\bR^4)$. Taking into account that $E\circ P=0$, we have associated to each element in $\mS^C(\bR^4)$ an equivalence class in $\frac{C^\infty_{tc,C}(\bR^4)}{P[C^\infty_{tc,C}(\bR^4)]}$. 
We focus now on injectivity. Let $\alpha\in C^\infty_{tc,C}(\bR^4)$ and let $\phi_\alpha\doteq E(\alpha)$ where $E$ is the causal propagator of $P$ on Minkowski spacetime. Per construction $P\phi_\alpha=0$. Furthermore since both the map $\iota_z:\bR^4\to\bR^4$ such that $\iota_z(\underline{x},z)=(\underline{x},-z)$ and $\iota_s:\bR^4\to\bR^4$ such that $\iota_s(\underline{x},z)=(\underline{x}, z+s)$, $s\in\bR$, are isometries of $(\bR^4,\eta)$ it holds that $E\circ\iota_z=\iota_z\circ E$ and $E\circ\iota_s=\iota_s\circ E$. Consequently $\phi=E(\alpha)=E(-\iota_z\alpha)=-\iota_z E(\alpha)=-\iota_z\phi$ which entails $\phi(\underline{x},0)=0$. At the same time, replacing $\iota_z$ with $\iota_s\circ\iota_z$, $s=2d$, we obtain that $\phi(\underline{x},2d-z)=-\phi(\underline{x},z)$ which implies $\phi(\underline{x},d)=0$. Since $E\circ P=0$, the map which associates to each $[\alpha]\in\frac{C^\infty_{tc,C}(\bR^4)}{P[C^\infty_{tc,C}(\bR^4)]}$, $E(\alpha)\in\mS^C(\bR^4)$ does not depend on the choice of the representative in $[\alpha]$ and it is, moreover, injective. As a matter of facts, suppose $E(\alpha)=0$. This entails that there exists $\rho\in C^\infty_{tc}(\bR^4)$ such that $\alpha=P\rho$. Yet, since $\alpha(\underline{x},z)=-\alpha(\underline{x},-z)=-\alpha(\underline{x},2d-z)$ and since $P$ is invariant both under the map $(\underline{x},z)\mapsto (\underline{x},-z)$ and $(\underline{x},z)\mapsto(\underline{x},z+2d)$, $\rho\in C^\infty_{tc,C}(\bR^4)$. As a consequence $P\rho$ lies in the trivial equivalence class of $\frac{C^\infty_{tc,C}(\bR^4)}{P[C^\infty_{tc,C}(\bR^4)]}$.
\end{proof}

\begin{remark}\label{image-conf}
It is noteworthy that the two conditions defining the elements of $\mS^C(\bR^4)$ in \eqref{image-sol} are actually already implementing the method of images at a level of dynamical configurations. As a matter of facts, consider any $\phi\in\mS^C(\bR^4)$: For any $n\in\bZ$, first applying the reflection along the hyperplane $(\underline{x},d)$ and then the one along $(\underline{x},0)$, the following chain of identities holds true:
$$\phi(\underline{x},z+2nd)=-\phi(\underline{x},-z-2(n-1)d)=\phi(\underline{x},z+2(n-1)d),$$
and equivalently $\phi(\underline{x},z+2nd)=\phi(\underline{x},z+2(n+1)d)$. In other words every element in $\mS^C(\bR^4)$ is both odd with respect to the reflection along the hyperplane $z=0$ and $2d$-periodic.\\
\end{remark}

Our next goal is to expand cohesively the content of the above remark. Therein our philosophy was to show that, to each dynamical configuration for a Casimir system, we can associate a solution of the equation of motion of a Klein-Gordon scalar field, which is periodic along the $z$-direction. From the quantum field theory point of view, especially when constructing algebraic states, we will be interested in a complementary problem, namely we would like to start from an element of $\mathcal{S}^{KG}(\bR^4)\doteq\{\phi\in C^\infty(\bR^4)\;|\;P\phi=0\}$ and associate to it one in $\mS^C(\bR^4)$. Following an argument almost identical to that of Proposition \ref{dynCas}, this problem can be translated to associating to an element of $C^\infty_{tc}(\bR^4)$ one of $C^\infty_{tc,C}(\bR^4)$. 
Barring the reflection along the plane $z=0$, the key procedure consists of making a smooth function on $\bR^4$ periodic. This operation, which is strongly tied to the Poisson's summation formula -- see \cite[\S 7.2]{Hormander1}, does not yield in general a well-defined result on the whole $C^\infty_{tc}(\bR^4)$. Yet we can individuate a notable subset which suffices to reach our goal. More precisely

\begin{proposition}\label{surjective}
Let $C^\infty_{0,C}(\bR^4)\doteq\{\alpha\in C^\infty_{tc,C}(\bR^4)\;|\;\supp(\alpha)\cap\left(\bR^3\times\{z\}\right)\;\textrm{is compact}\;\forall z\in\bR\}$ and let $N: C^\infty_{0}(\bR^4)\to C^\infty_{0,C}(\bR^4)$ be defined as 
\begin{equation}\label{imagesum}
N(f)(\underline{x},z)=\sum\limits_{n=-\infty}^\infty \left(f(\underline{x},z+2nd)-f(\underline{x},-z+2nd)\right).
\end{equation}
The following statements hold true:
\begin{enumerate}
\item The map $N$ is surjective, but not injective.
\item $N$ is an isomorphism between $C^\infty_0(\mathring{Z})\subset C^\infty_0(\bR^4)$ and $C^\infty_{0,I}(\bR^4)\subset C^\infty_{0,C}(\bR^4)$, where 
$$C^\infty_{0,I}(\bR^4)\doteq\{\alpha\in C^\infty_{tc,C}(\bR^4)\;|\;\supp(\alpha)\cap \mathring{Z}\;\textrm{is compact}\}.$$
\end{enumerate}
\end{proposition}

\begin{proof}
We remark, that, per construction $N(f)$ is a smooth function which is $2d$-periodic and odd for reflection along the $z$-axis for any $f\in C^\infty_0(\bR^4)$. The compact support ensures the convergence of the series. 

Let us focus on {\em 1.}: To show that $N$ is surjective, let $\zeta\in C^\infty_{0,C}(\bR^4)$ and let $\chi\in C^\infty(\bR^4)$ be a function constructed as follows. It depends only on $z$ and, at fixed value of $\underline{x}$, $\chi(z)\in C^\infty_0(\bR^4)$ in such a way that $\chi$ vanishes for all $|z|\geq 2d-\alpha$, $\alpha\in (0,d)$. Furthermore $\chi(z)=1$ if $z\in (-\alpha,\alpha]$ and for all other values of $z$ it is such to satisfy the identity $\chi(z)+\chi(z+2d)=1$ for all $z\in [-2d,0]$. Consequently $\chi\zeta\in C^\infty_0(\bR^4)$ and a direct calculation shows that $N(\chi\zeta)=\zeta$. Hence $N$ is surjective. To show that $N$ is not injective it suffices to exhibit an explicit example: Consider any $\beta\in C^\infty_0\left((0,d)\times\bR^3\right)$ and 
$f(\underline{x},z)$ as $\beta(\underline{x},z)$ if $z>0$ and as $-\beta(\underline{x},-z)$ if $z<0$. At the same time define $$\beta^\prime(\underline{x},z)=\left\{\begin{array}{ll}
\frac{1}{2}\beta(\underline{x},z) & z\in (0,d)\\
\frac{1}{2}\beta(\underline{x},z-d) & z\in (d,2d)
\end{array}\right. .$$
If we consider $f^\prime(\underline{x},z)$ as $\beta^\prime(\underline{x},z)$ if $z>0$ and as $-\beta^\prime(\underline{x},-z)$ if $z<0$, using \eqref{imagesum}, it turns out that $N(f)=N(f^\prime)$. 

Let us now focus on {\em 2.}: Let $\zeta\in C^\infty_{0,I}(\bR^4)$; per definition $f\doteq\left.\zeta\right|_{Z}\in C^\infty_0(Z)$. On account of \eqref{imagesum} $N(f)=\zeta$, that is $N$ is surjective on $C^\infty_{0,C}(\bR^4)$. Let us assume that there exists $f^\prime\in C^\infty_0(Z)$ such that $N(f^\prime)=0$. Formula \eqref{imagesum} entails that $\left.N(f^\prime)\right|_{Z}=f^\prime=0$, which proves that $N$ is injective. 
\end{proof}

According to our overall strategy, the next step calls for the identification of a counterpart for a Casimir system of ${E}_{\bH^4}$ which played a key role in studying a Casimir-Polder system. Notice that the key role of ${E}_{\bH^4}$ was on the one hand to generate all smooth solutions with the wanted boundary conditions, while on the other hand, it yielded a symplectic form on the space of classical observables. We have emphasized this second aspect since it is easy to grasp that identifying eventually a symplectic form in Casimir system, is more difficult on account of the periodicity of the elements in $\mS^C(Z)$. A solution to this problem lies in this proposition:

\begin{proposition}\label{sc-solutions}
We call $\mS^C_{sc}(Z)$ the collection of all solutions $u\in C^\infty(Z)$ of \eqref{Casimir-eom} such that $\supp(u)\cap\left(\{t\}\times\bR^2\times [0,d]\right)$ is compact for all $t\in\bR$. This is  
\begin{enumerate}
\item a vector space isomorphic to $\frac{C^\infty_{0,C}(\bR^4)}{P[C^\infty_{0,C}(\bR^4)]}$,
\item a symplectic space if endowed with the following weakly non-degenerate symplectic form:
\begin{equation}\label{E-sympl}
\sigma_C(u,u^\prime)=\sigma_C([\zeta],[\zeta^\prime])=\left(\zeta,E(\zeta^\prime)\right)_C=-\left(E(\zeta),\zeta^\prime\right)_C,
\end{equation}
where $\zeta$ and $\zeta^\prime$ are representatives of $[\zeta],[\zeta^\prime]\in \frac{C^\infty_{0,C}(\bR^4)}{P[C^\infty_{0,C}(\bR^4)]}$ so that $u=E(\zeta)$ and $u^\prime = E(\zeta^\prime)$ and where 
\begin{equation}\label{E-sympl2}
\left(\zeta,E(\zeta^\prime)\right)_C\doteq\int\limits_{\bR^3}d^3\underline{x}\int\limits_{0}^d dz\, \zeta E(\zeta^\prime)=
-\int\limits_{\bR^3}d^3\underline{x}\int\limits_{0}^d dz\, E(\zeta) \zeta^\prime.
\end{equation}
\end{enumerate}
\end{proposition}

\begin{proof}
On account of Proposition \ref{dynCas}, to every element $u\in\mS^C_{sc}(Z)\subset\mS^C(Z)$, we can associate via the map $F$ in \eqref{image-sol} a function $\phi\in C^\infty_{-}(\bR^4)$, solution of $P\phi=0$, so that $u=\left.\phi\right|_Z$. Furthermore, there exists  $[\alpha]\in\frac{C^\infty_{tc,C}(\bR^4)}{P\left[C^\infty_{tc,C}(\bR^4)\right]}$ such that $\phi=E(\alpha)$. Since $\phi$ is per hypothesis compactly supported along the $x,y$-directions, but neither in time nor along $z$, the standard support properties of the causal propagator $E$ entail, in turn, that $\alpha$ must be smooth and compactly supported along the $t,x,y$-directions without additional constraints imposed along the $z$-direction. Repeating slavishly the proof of Proposition \ref{dynCas}, {\em 1.} descends. 

Let us focus on {\em 2.}: As a first step, we show that \eqref{E-sympl2} is well-posed. Since for any $u,u^\prime\in\mS^C(Z)$, there exists $[\zeta],[\zeta^\prime]\in\frac{C^\infty_{0,C}(\bR^4)}{P[C^\infty_{0,C}(\bR^4)]}$, such that $u^{(\prime)}=E(\zeta^{(\prime)})$, well-posedness descends from showing that for any $\zeta^\prime\in C^\infty_{0,C}(\bR^4)$ it vanishes the integral
$$\int\limits_{\bR^3}d^3\underline{x}\int\limits_{0}^d dz\, P(\zeta^\prime) E(\zeta).$$
Define $P_{(3)}=P-\frac{\partial^2}{\partial z^2}$ and rewrite the integral as 
$$\int\limits_{\bR^3}d^3\underline{x}\int\limits_{0}^d dz\,\left(P_{(3)}+\frac{\partial^2\zeta^\prime}{\partial z^2}\right) E(\zeta)=
\int\limits_{\bR^3}d^3\underline{x}\int\limits_{0}^d dz\,\left( \frac{\partial^2\zeta^\prime}{\partial z^2} E(\zeta)+\zeta^\prime P_{(3)}E(\zeta)\right),$$
where we used both that $P_{(3)}$ is a formally self-adjoint operator which does not depend on $z$ and that we are integrating along the whole $\bR^3$. If we use the identity $P_{(3)} E(\zeta)= PE(\zeta)-\frac{\partial^2 E(\zeta)}{\partial z^2}$ and we integrate by parts, it holds
$$\int\limits_{\bR^3}d^3\underline{x}\int\limits_{0}^d dz\,\left( \frac{\partial^2\zeta^\prime}{\partial z^2} E(\zeta)+\zeta^\prime P_{(3)}E(\zeta)\right)=\left.\left(\frac{\partial\zeta^\prime}{\partial z}E(\zeta)-\zeta^\prime\frac{\partial E(\zeta)}{\partial z}\right)\right|_{0}^d=0,$$
where we used that both $\zeta^\prime$ and $E(\zeta)$ vanish both at $z=0$ and at $z=d$. From this computation it also descends that, for any $\zeta, \zeta^\prime\in C^\infty_{0,C}(\bR^4)$  
$$(\zeta,E^+(\zeta^\prime))_C = (P E^-\zeta, E^+(\zeta^\prime))_C=(E^-\zeta, \zeta^\prime)_C,$$
where $E^\pm$ are the advanced and the retarded fundamental solutions of $P$ on the whole Minkowski spacetime. From this last identity it descends that $(\zeta, E(\zeta^\prime))_C=-(E(\zeta),\zeta^\prime)_C$. In other words, $\sigma_C$ is both bilinear and antisymmetric. To prove non-degenerateness, suppose there exists $[\zeta]\in\frac{C^\infty_{0,C}(\bR^4)}{P[C^\infty_{0,C}(\bR^4)]}$ such that 
$(\zeta,E(\zeta^\prime))_C=0$ for all $[\zeta^\prime]\in\frac{C^\infty_{0,C}(\bR^4)}{P[C^\infty_{0,C}(\bR^4)]}$. In particular this entails that $(E(\zeta),\zeta^\prime)_C=0$. If we choose $\zeta^\prime$ so that $\left(\supp(\zeta^\prime)\cap Z\right)\subset\mathring{Z}$, calling $\zeta_0\doteq\left.\zeta^\prime\right|_Z$ the following identity holds true:
$$(E(\zeta),\zeta^\prime)_C=\int\limits_{\bR^4}d^4x\, E(\zeta) \zeta_0.$$
Notice that $\zeta\in C^\infty_{0,C}(\bR^4)\subset C^\infty_{tc}(\bR^4)$ and that the right hand side coincides with the standard pairing between $\frac{C^\infty_{tc}(\bR^4)}{P[C^\infty_{tc}(\bR^4)]}$ and $\frac{C^\infty_0(\bR^4)}{P[C^\infty_0\bR^4)]}$ on the whole Minkowski spacetime, which is non degenerate -- see for example \cite{Benini}. Hence there must exist $\alpha\in C^\infty_{tc}(\bR^4)$ such that $\zeta=P\alpha$. Notice that, since \eqref{imagesum} guarantees us that $N$ is built out of isometries of the standard Minkowski metric, it holds that $E^\pm\circ N = N\circ E^\pm$. Since $\alpha=E^+(\zeta)=E^-(\zeta)$ and $\zeta\in C^\infty_{0,C}(\bR^4)$ per hypothesis, $\alpha$ lies in $C^\infty_{0,C}(\bR^4)$, concluding the proof that $\sigma_C$ is weakly non-degenerate.
\end{proof}

Notice that restricting the domain of integration in \eqref{E-sympl2} is necessary to obtain finite quantities and it encodes the physical idea that only the information contained between the boundaries at $z=0$ and at $z=d$ are physically relevant. Before concluding this part of our investigation of a Casimir system, we elaborate from Proposition \ref{sc-solutions} the following Definition

\begin{definition}\label{Casimirpropagator}
We call {\bf Casimir causal propagator} the map $E_Z:C^\infty_0(\bR^4)\to\mathcal{S}_{sc}^C(Z)$ where 
\[
E_Z\doteq \rho_Z\circ E\circ N,
\]
where $N$ is defined in \eqref{imagesum}, $E$ is the causal propagator of the Klein-Gordon scalar field on Minkowski spacetime, while $\rho_Z$ is the restriction map to $Z$.
\end{definition}

\begin{remark}\label{nogamenolife} 
Notice that there is no symplectic isomorphism between $\mathcal{S}_{sc}^C(\bR^4)$ and the space of spacelike compact solutions of the Klein-Gordon equation on Minkowski spacetime . The reason is that $N$ does not preserve the symplectic form, since for arbitrary $f,f^\prime\in C^\infty_0(\bR^4)$,
\begin{equation}\label{tantecarecose}
E(f,f^\prime)\neq\sigma_C\left((E\circ N)(f), (E\circ N)(f^\prime)\right)=(\zeta,E(\zeta^\prime))_C,
\end{equation}
where $\sigma_C$ is the one introduced in \eqref{E-sympl} and, setting $\zeta=N(f)$ and $\zeta^\prime=N(f^\prime)$, the last equality holds on account of \eqref{E-sympl2}. The main consequence of this failure will be the impossibility at a later stage to construct states for the algebra of observables for a Casimir system as the pull-back of states for the counterpart on the whole Minkowski spacetime.
\end{remark}

\vskip .2cm

\noindent\textbf{\em Part 2 -- The off-shell algebra:}  Having characterized all possible dynamical configurations for a Casimir system, we can address the question on how to build an algebra of observables following the example given in Section \ref{notations}. Our guiding principle will be the same as in section \ref{CPsystem} and, in particular, we shall use the functional formalism. We stress that there will be several modifications in comparison to our analysis of the previous section. These can be ultimately ascribed to the more complicated underlying geometry and to the fact that we have well-under control the convergence of the series \eqref{imagesum} only with respect to compactly supported functions. 

\begin{definition}\label{Coff}
We call {\em space of kinematical/off-shell configurations} for a Casimir system 
$$\mathcal{C}^{C}(Z)\doteq\left\{u\in C^\infty(Z)\;|\;\left.u\right|_{\partial Z}=0\;\textrm{and}\;\exists \phi\in \mathcal{C}^{KG}(\bR^4)\;\textrm{such that}\;u=\left.\phi\right|_Z\right\},$$
We consider $\mathcal{C}^{C}(Z)$ endowed with the compact-open topology.
\end{definition}

Notice that $\mathcal{S}^C(Z)\subset\mathcal{C}^C(Z)$. As next step, we want to construct a space of functionals measuring off-shell configurations and we want to endow it with the structure of a $*$-algebra. In this respect Definition \ref{Casimirpropagator} plays a key role.  
\begin{definition}\label{functionalC}
Let $F:\mathcal{C}^{C}(Z)\to\bC$ be any smooth functional. We call it {\bf regular} if for all $k\geq 1$ and for all $u\in\mathcal{C}^{C}(Z)$, $F^{(k)}[u]\in C^{\infty}_{0,C}(\mathring{Z}^k)$, and if only finitely many functional derivatives do not vanish. We indicate this set as $\mathcal{F}_0^C(Z)$.
\end{definition}

\noindent Let us consider the following map:
$$\star_Z:\mathcal{F}_0^C(Z)\times\mathcal{F}_0^C(Z)\to\mathcal{F}_0^C(Z),$$
which associates to each $F,F^\prime\in\mathcal{F}_0^C(Z)$
\begin{equation}\label{starN}
\left(F\star_ZF^\prime\right)(u)=\left(\mathcal{M}\circ\exp(i\Gamma_{E_Z})(F\otimes F^\prime)\right)(u)
\end{equation}
Here $\mathcal{M}$ stands for the pointwise multiplication, {\it i.e.}, $\mathcal{M}(F\otimes F^\prime)(u)\doteq F(u)F^\prime(u)$, whereas 
$$
\Gamma_{E_Z}\doteq\frac{1}{2}\int_{Z\times Z}E_Z(x,x^\prime)\frac{\delta}{\delta u(x)}\otimes\frac{\delta}{\delta u(x^\prime)},
$$
where $E_Z(x,x^\prime)$ is the integral kernel of \eqref{modified}. The exponential in \eqref{starN} is defined intrinsically in terms of the associated power series and, consequently, we can rewrite the product also as 
\begin{equation}\label{algprod5}
\left( F\star_Z F^\prime\right)(u)=\sum\limits_{n=0}^\infty\frac{i^n}{2^n n!}\langle F^{(n)}(u),E_Z^{\otimes n}(F^{\prime (n)})(u)\rangle,
\end{equation}
where the $0$-th order is defined as the pointwise multiplication, that is $\langle F^{(0)}(u),F^{\prime (0)}(u)\rangle\doteq F(u)F^\prime(u)$. Notice that \eqref{algprod5} is well-defined, since $E_Z=E\circ N$ and thus elements in $C^\infty_{0,C}(Z)$ are per definition such that their image under the action of $N$ lies in $C^\infty_{tc,C}(Z)$. To summarize

\begin{definition}
We call $\mathcal{A}^{C}(Z)\equiv\left(\mathcal{F}_0^C(Z),\star_Z\right)$ the {\em off-shell $*$-algebra} of a Casimir system endowed with complex conjugation as $*$-operation. 
\end{definition} 

\begin{remark}\label{PolC}
Notice that, in complete analogy with $\mathcal{A}^{CP}(\bH^4)$ and with $\mathcal{A}^{KG}(\bR^4)$, $\mathcal{A}^C(Z)$ can be seen as being generated by the functionals $F_h(u)=\int\limits_{\bR^3}d^3\underline{x}\int\limits_0^d dz\, u(\underline{x},z)h(\underline{x},z)$ where $h\in C^\infty_{0,C}(Z)$, while $u\in\mathcal{C}^C(Z)$. At the same time, if we consider as generating functionals only those whose labeling space is $C^\infty_{0,I}(Z)$, we obtain the {\em extensible $*$-algebra} $\mathcal{A}^{C}_{ext}(Z)$, which is a $*$-subalgebra of both $\mathcal{A}^{C}(Z)$. Notice that $\mathcal{A}^{C}_{ext}(Z)$ plays a distinguished role as we will be able to define Hadamard states only for such algebra.
\end{remark}

The causal propagator for the Casimir system is constructed modifying the causal propagator of the Minkowski spacetime with the operator $N$. 
Thanks to the causal properties of $E$, when employed on test functions supported in a globally hyperbolic set $O$ strictly contained in $Z$ it holds that 
\[
E_Z(f,f') =  E(f,N(f'))   ,\qquad  f,f'\in C^\infty_0(O)
\]
because the reflections and the translations used in $N$ map the support of $f'$ in regions which are causally disjoint from $O$. The following proposition states that the local algebra of observables of a Casimir system cannot be distinguished from Klein-Gordon counterpart.  

\begin{proposition}\label{localcomparison-c}
Let $O$ be any globally hyperbolic open region strictly contained in $Z$. There exists a $*$-isomorphism between $\mathcal{A}^{KG}(O)\doteq\left.\mathcal{A}^{KG}(\bR^4)\right|_{O}$ and $\mathcal{A}^{C}(O)\doteq\left.\mathcal{A}^{C}(Z)\right|_O$. The isomorphism is implemented by the identity map.
\end{proposition}

The proof of this proposition can be obtained along the guidelines of that of Proposition \ref{timeslice} together with the property of $E_Z$ stated above. 

\vskip .2cm

\noindent\textbf{\em Part 3 -- The on-shell algebra:} Having investigated the algebra probing kinematical configurations, we want to conclude our analysis by constructing the counterpart on the solutions to the equation of motion. This is tantamount to restricting the allowed configurations from $\mathcal{C}^{C}(Z)$ to $\mathcal{S}^{C}(Z)$. As outlined in Section \ref{notations} and in Section \ref{CPsystem} for a Casimir-Polder system, this entails that several functionals become redundant as they are automatically vanishing when evaluated on any solution. This calls for the identification and for the elimination of these observables via a suitable quotient. At a level of algebras the solution of this problem is contained in Proposition \ref{dynCas} and in the isomorphism between $\mS^C(Z)$ and $\frac{C^\infty_{tc,C}(\bR^4)}{P[C^\infty_{tc,C}(\bR^4)]}$. This suggests to consider the functionals $F_{[\zeta]}:\frac{C^\infty_{tc,C}(\bR^4)}{P[C^\infty_{tc,C}(\bR^4)]}\to\bR$ so that $F_{[\zeta]}([\alpha])=(\zeta,E(\alpha))_C$, where the right hand side is defined in \eqref{E-sympl2}. 

Notice that, still in view of Proposition \ref{dynCas}, we can rewrite each of these functionals also as $F_{[\zeta]}:\mS^{C}(Z)\to\bC$, thus as a genuine classical observable on the dynamical configurations of a Casimir system. The underlying philosophy is to single out via the labeling space $\frac{C^\infty_{0,C}(\bR^4)}{P[C^\infty_{0,C}(\bR^4)]}$ the generators of an on-shell algebra of observables for a Casimir system. As a preliminary step, we exhibit some relevant properties of these generating functionals, which justify their choice:

\begin{proposition}\label{obsC}
We call {\bf classical observable for a Casimir system} the map $F_{[\zeta]}:\frac{C^\infty_{tc,C}(\bR^4)}{P[C^\infty_{tc,C}(\bR^4)]}\to\bC$, $[\zeta]\in\frac{C^\infty_{0,C}(\bR^4)}{P[C^\infty_{0,C}(\bR^4)]}$, defined as
\begin{equation}\label{classicalobservable}
F_{[\zeta]}([\alpha])=(\zeta,E(\alpha))_C,
\end{equation}
where $\zeta$ and $\alpha$ are arbitrary representatives of $[\zeta]$ and $[\alpha]$ respectively. The collection of all classical observables $\mO^C(Z)$ is a vector space which is both separating and optimal in the sense of Proposition \ref{obsCP}. Furthermore $(\mO^C(Z),\sigma_C)$ is a symplectic space, $\sigma_C$ being defined in \eqref{E-sympl}.
\end{proposition} 

\begin{proof}
We notice that \eqref{classicalobservable} is a well-defined quantity whose right hand side does not depend on the representatives chosen, as one can infer by repeating slavishly the same reasoning as in Proposition \ref{sc-solutions} using additionally that $\left(\supp(\zeta)\cap\supp(E(\alpha)\right)\cap Z$ is compact.

Since $\mO_{[\zeta]}$ is linear in $[\zeta]$, $\mO^C(Z)$ is a vector space which is isomorphic to $\mS_{sc}^C(Z)$. Hence, since the latter is a symplectic space as proven in Proposition \ref{sc-solutions}, so is $\mO^C(Z)$ endowed with $\sigma_C$. We need only to show that the collection of classical observables is separating and optimal. The first descends from the following remark:$\frac{C^\infty_{tc,C}(\bR^4)}{P[C^\infty_{tc,C}(\bR^4)]}$ is isomorphic via $E$ to $\mS^C(Z)$ which in turn identifies a vector subspace of $\frac{C^\infty(\bR^3\times(0,d))}{P[C^\infty(\bR^3\times(0,d))]}$. With respect to the pairing we have introduced, standard 
arguments in functional analysis guarantee that $\frac{C^\infty_0(\bR^3\times(0,d))}{P[C^\infty_0(\bR^3\times(0,d))]}$ separates $\frac{C^\infty(\bR^3\times(0,d))}{P[C^\infty(\bR^3\times(0,d))]}$. Since $C^\infty_0(\bR^3\times(0,d))\subset C^\infty_{0,C}(Z)$ the sought statement holds true. 

To conclude we show that our choice is optimal. Suppose that there exists a classical observable generated by $\zeta\in C^\infty_{0,C}(\bR^4)$ such that $(\zeta,E(\alpha))_C=0$ for all $\alpha\in C^\infty_{tc,C}(\bR^4)$. Equivalently this entails that $(E(\zeta),\alpha)_C=0$. Since $\alpha$ is an arbitrary timelike compact function in $\bR^3\times (0,d)$, the same reasoning as for the scalar field on the whole Minkowski spacetime entails that $E(\zeta)$ must vanish thereon. In other words $\zeta\in P[C^\infty_{0,C}(\bR^4)]$, that is it generates the trivial class in $\frac{C^\infty_{0,C}(\bR^4)}{P[C^\infty_{0,C}(\bR^4)]}$. 
\end{proof} 

\noindent We have all the ingredients to introduce the following structure:

\begin{definition}\label{Calgebra}
We call {\bf on-shell $*$-algebra of observables for a Casimir system} the algebra $\left(\mathcal{A}^C_{on}(Z),\star_Z\right)$ generated by the functionals $F_{[\zeta]}:\mathcal{S}^C(Z)\to\bC$ with $[\zeta]\in\frac{C^\infty_{0,C}(\bR^4)}{P[C^\infty_{0,C}(\bR^4)]}$ such that $F_{[\zeta]}(u)=\int\limits_{\bR^3}d^3\underline{x}\int\limits_0^d dz\, \zeta(\underline{x},z)u(\underline{x},z)$, $u\in\mathcal{S}^C(Z)$.
\end{definition}

\noindent Let us show that our choice for the algebra of observables enjoys notable properties:

\begin{lemma}
The algebra $\mathcal{A}^C_{on}(Z)$ is causal and it satisfies the time-slice axiom.
\end{lemma}

\begin{proof}
The property of an algebra being causal is tantamount to showing that spacelike separated observables do commute. It suffices to check it for all generators and it is equivalent to proving that, for all $[\zeta], [\zeta^\prime]\in \frac{C^\infty_{0,C}(\bR^4)}{P[C^\infty_{0,C}(\bR^4)]}$, it holds $\sigma_C([\zeta],[\zeta^\prime])=0$ if there exists two representative $\zeta,\zeta^\prime$ which are spacelike separated. On account of Proposition \ref{sc-solutions} this is a consequence of the support properties of the causal propagator. 

With respect to the time-slice axiom, mutatis mutandis, the procedure is identical to the one outlined in the proof of Lemma \ref{timeslice} and we shall thus not repeat it.
\end{proof}

To conclude we remark that $\mathcal{A}^{C}_{on}(Z)$ could have been realized also as the quotient between $\mathcal{A}^C(Z)$ and the $*$-ideal generated by elements of the form $Ph$, where $P$ is the Klein-Gordon operator and $h\in C^\infty_{0,C}(Z)$.

\subsection{Hadamard states for a Casimir system}

In this section we discuss a possible way to construct a certain class of states for the Casimir system. 
We shall restrict our attention to those which are quasi-free and have suitable regularity. In particular we follow the same philosophy used in the previous section, namely  we will focus our attention on those states from which stems a prescription to construct Wick polynomials which coincides with the standard one if we restrict our attention to any globally hyperbolic submanifold $O\subset Z$. Well-posedness of this line of thought is a by-product of Proposition \ref{localcomparison-c}, which guarantees that $\mathcal{A}^{C}(O)$ is $*$-isomorphic to $\mathcal{A}^{KG}(O)$. Accordingly,

\begin{definition}\label{def:hadamard-casimir}
A state $\omega:\mathcal{A}^C(Z)\to\bC$ is of {\bf Hadamard} form if it is normalized, positive, quasi-free and, if, for any globally hyperbolic submanifold $O\subset Z$, the restriction of $\omega$ to $\mathcal{A}^C(O)$ is such that there exists $\omega_2\in\mathcal{D}^\prime(O\times O)$ whose wavefront set is 
$$WF(\omega_2)=\left\{(x,x^\prime,k_x,-k_{x^\prime})\in T^*(O\times O)\setminus \{\mathbf 0\}\;|\;
(x,k_x)\sim(x^\prime,k_{x^\prime}),\;k_x\triangleright 0\right\}, $$
 and, for all $F_{h},F_{h^\prime}\in\mathcal{A}^{C}(O)$
$$
\omega\left(F_h\star_{Z} F_{h^\prime}\right)=\omega_2(h,h^\prime),\quad h, h^\prime \in C^\infty_0(O).
$$
\end{definition}

\vskip .2cm

As for a Casimir-Polder system we want to exhibit explicit examples of Hadmard states for a Casimir system and our initial plan is to build them starting from a quasi-free counterpart $\widetilde\omega:\mathcal{A}^{KG}(\bR^4)\to\bC$, which is of Hadamard form itself. In other words we would like to mimic the content of Proposition \ref{comparison}. Alas, there does not exist a $*-$homomorphism between $\mathcal{A}^C(Z)$ and $\mathcal{A}^{KG}(\bR^4)$ and hence no corresponding pull-back of states. We shall avoid such hurdle by working directly at the level of the two-point function adapting the {\bf image method} used previously in Definition \ref{Casimirpropagator} for the causal propagator. Notice that, with this procedure, we will be constructing actually a state for $\mathcal{A}^C_{ext}(Z)$.

More precisely our starting point is any Hadamard state $\widetilde\omega:\mathcal{A}^{KG}(\bR^4)\to\bC$, whose associated two-point function $\widetilde\omega_2\in\mathcal{D}^\prime(\bR^4\times\bR^4)$. In view of Definition \ref{Casimirpropagator}, applying the image method to $\widetilde\omega_2$ is tantamount to proving that $\widetilde\omega_2\circ (\mathbb{I}\otimes N)\in\mathcal{D}^\prime(\mathring{Z}\times \mathring{Z})$. Notice that the outcome does not define an {\em image state} for $\mathcal{A}^C(Z)$ but only for $\mathcal{A}^C_{ext}(Z)$.

Since our goal is to exhibit explicit cases where this procedure works, we restrict the attention only to quasi-free states for $\mathcal{A}^{KG}(\bR^4)$ whose associated two-function has an integral kernel which is invariant under the simultaneous action on both entries of both $\iota_z$, the reflection along the hyperplane $z=0$ and of $\iota_s$, the translation of step $s$ along the $z$-direction, $s\in\bR$:
\begin{equation}\label{eq:symmetry-casimir}
\widetilde{\omega}_2(\iota_z(f),\iota_z(f^\prime)) = \widetilde{\omega}_2(\iota_s(f),\iota_s(f^\prime))= \widetilde{\omega}_2(f,f^\prime),
\end{equation}
where $f,f^\prime\in C^\infty_0(\bR^4)$. As an additional ingredient we recall, that all two-points functions of Hadamard form differ only by a smooth integral kernel. Hence, since in this section we are interested in a massless real scalar field, we can split
\begin{equation}\label{eq:hadamard-casimir}
\widetilde{\omega}_2(x,x^\prime)=\widetilde{\omega}_2^0(x,x^\prime)+W(x,x^\prime)
\end{equation}
where $W\in C^\infty(\bR^4\times\bR^4)$, while $\widetilde\omega_2^0(x,x^\prime)$ is the integral-kernel of the two-point function of the Poincar\'e vacuum. Therefore, we will analyze separately $W(x,x^\prime)$ and $\widetilde\omega_2^0(x,x^\prime)$ starting from the latter, which fulfills the requirements of \eqref{eq:symmetry-casimir}. Recall that
\begin{equation}\label{eq:2pt-mink}
\widetilde{\omega}_2^0(f,f^\prime) \doteq \lim_{\epsilon\to 0^+} \frac{1}{4\pi^2}\int\limits_{\bR^4\times\bR^4} d^4x\, d^4x^\prime   \frac{f(x) f^\prime(x^\prime)}{-(x^0-x^{\prime 0}-i\epsilon)^2+(x^1-x^{\prime 1})^2+(x^2-x^{\prime 2})^2+(x^3-x^{\prime 3})^2}.
\end{equation}
Upon Fourier transform, we can rewrite the last expression as
\begin{equation}\label{eq:2pt-ft}
\widetilde{\omega}_2^0(f,f^\prime) =\int_\bR dt \int_\bR dt' \int_{\bR^3} d^3{\bf k}   \frac{1}{2|{\bf k}|}   e^{-i (t-t')|{\bf k}|}    \widehat{f}(t,{\bf k}) \overline{\widehat{f}^\prime(t',{\bf k})}
\end{equation}
where $\widehat{f}(t,{\bf k})$ is the three dimensional spatial Fourier transform\footnote{Our convention for the spatial Fourier transform is the following: $\widehat{f}(t,{\bf k}) \doteq \frac{1}{(\sqrt{2\pi})^3} \int d^3 {\bf x} e^{i{\bf k}\cdot x}f(t,{\bf x}) $.} of $f(t,{\bf x})$.

\begin{proposition}\label{PoincN}
Let $\widetilde{\omega}^0_2$ be the two-point function of the Poincar\'e vacuum for a real, massless scalar field on Minkowski spacetime.
Then $\omega^0_2\doteq\widetilde\omega_2^0\circ\left(\mathbb{I}\otimes N\right)=\widetilde\omega_2^0\circ\left(N\otimes\mathbb{I}\right)\in\mathcal{D}^\prime(\mathring{Z}\times\mathring{Z})$. Furthermore the integral kernel of $\omega^0_2$ can be written as the $\epsilon\to 0$ limit of the following $\epsilon$-regularized integral kernel:
\begin{equation}\label{eq:2pt-c0-limit}
\frac{1}{8\pi d\chi_\epsilon}\left(\frac{\sinh\frac{\pi\chi_\epsilon}{d}}{\cosh\frac{\pi\chi_\epsilon}{d}-\cos\left(\frac{\pi}{d}(z-z^\prime)\right)}-
\frac{\sinh\frac{\pi\chi_\epsilon}{d}}{\cosh\frac{\pi\chi_\epsilon}{d}-\cos\left(\frac{\pi}{d}(z+z^\prime)\right)}\right)
\end{equation}
where $\chi_\epsilon \doteq -(\underline x^0-\underline x^{\prime 0}-i\epsilon)^2 + (\underline x^1-\underline x^{\prime 1})^2 + (\underline x^2-\underline x^{\prime 2})^2$.
\end{proposition}
\begin{proof}
With respect to the standard Cartesian coordinates (but keeping the notation $x=(\underline x,z)$) and fixing $e^0 = (1,0,0,0)$ and $e^3= (0,0,0,1)$, we can write the formal expression
\begin{gather*}
\widetilde{\omega}_2^0(f,Nf^\prime) = \lim_{\epsilon\to 0^+} \int d^4x d^4x^\prime 
\left[\widetilde{\omega}_2^0(x+i\epsilon e^0,x^\prime) - \widetilde{\omega}_2^0(x+i\epsilon e^0,\iota_z x^\prime)\right] f(x) \sum_{n} f^\prime(x^\prime+2nd e^3).
 \end{gather*}
 
Up to a change of variables of integration for every element of the sum, we obtain 
\[
\lim_{\epsilon\to 0^+} \int  d^4x d^4x^\prime
 \sum_{n} \left[\widetilde{\omega}_2^0(x+i\epsilon e^0,x^\prime+2dn e^3) - \widetilde{\omega}_2^0(x+i\epsilon e^0,\iota_z x^\prime+2dn e^3)\right] f(x) f^\prime(x^\prime).
\] 
For every $\epsilon>0$,  
\[
\left(\widetilde{\omega}_2^0\circ\left(\mathbb{I}\otimes N\right)\right)(x+i\epsilon e^0,x^\prime) = \lim_{m\to\infty}\sum_{|n|<m} \left[\widetilde{\omega}_2^0(x+i\epsilon e^0,x^\prime+2dn e^3) - \widetilde{\omega}_2^0(x+i\epsilon e^0,\iota_z x^\prime+2dn e^3)\right].
\]
If we recall that for complex variables $a,b\in\bC$, it holds -- see \cite[\S 1.445]{tableofintegrals} $$\sum\limits_{n=-\infty}^\infty\frac{1}{a^2+(b+n)^2}=\frac{\pi}{a}\frac{\sinh(2\pi a)}{\cosh(2\pi a)-\cos(2\pi b)},$$ and if we recall the form of $\widetilde{\omega}_2^0(x,x^\prime)$ given in \eqref{eq:2pt-mink} we can show that the $\widetilde{\omega}_2^0\circ \left(\mathbb{I}\otimes N\right)$  converges to
\[
\omega_2^0(x,x^\prime)=\lim_{\epsilon\to0^+}\frac{1}{8\pi d\chi_\epsilon}\left(\frac{\sinh\frac{\pi\chi_\epsilon}{d}}{\cosh\frac{\pi\chi_\epsilon}{d}-\cos\left(\frac{\pi}{d}(x^3-y^3)\right)}-
\frac{\sinh\frac{\pi\chi_\epsilon}{d}}{\cosh\frac{\pi\chi_\epsilon}{d}-\cos\left(\frac{\pi}{d}(x^3+y^3)\right)}\right).
\]
in the limit of $n\to \infty$. 

We interpret $\underline x^0-\underline x^{\prime 0}+i\epsilon$ as an extension of $\underline x^0-\underline x^{\prime 0}$ to the complex plane and we investigate the properties of \eqref{eq:2pt-c0-limit} as an analytic function. Notice that $\sinh(\xi)/\xi$ is entire analytic as a function of $\xi^2$. Hence its composition with $\xi^2=(\pi/d)^2\chi^2(\underline x,\underline x^\prime)=(\pi/d)^2(-(\underline x^0-\underline x^{\prime 0})^2 + (\underline x^1-\underline x^{\prime 1})^2 + (\underline x^2-\underline x^{\prime 2})^2)$ is in turn entire analytic itself on $\mathbb{C}^8$.  Furthermore, since the function $1/(\cosh(\alpha)-\cos(\beta))$ can be expanded in Laurent series in terms of  $\alpha^2$ and $\beta^2$ whenever $\cosh(\alpha) \neq \cos(\beta)$, this result applies to our scenario whenever $\underline x^0-\underline x^{\prime 0}+i\epsilon$ has a sufficiently large imaginary component while the other coordinates have a small imaginary part. Under these conditions we can conlcude the existence of a domain of analyticity for \eqref{eq:2pt-c0-limit}. Notice that a boundary component of such domain is obtained constraining all spatial coordinates to be real and taking the limit $\epsilon=\Im(\underline x^0-\underline x^{\prime 0})$ to $0^+$. Furthermore, by direct inspection, \eqref{eq:2pt-c0-limit} is bounded up to a multiplicative constant by $\epsilon^{-2}$, close to the mentioned boundary component. Hence we can apply Theorem 3.1.15 of H\"ormander \cite{Hormander1} to conclude that the boundary value of \eqref{eq:2pt-c0-limit} at $\epsilon=0$ is itself a distribution.
\end{proof}

\noindent To conclude that $\omega^0_2(x,x^\prime)$ defines a state on $\mathcal{A}^C_{ext}(Z)$ we prove the following:

\begin{proposition}\label{pr:vacuum-positivity}
The distribution ${\omega}_2^0\in\mathcal{D}^\prime(\mathring{Z}\times\mathring{Z})$ built in Proposition \ref{PoincN} is the two-point function of a quasi-free state $\omega^0:\mathcal{A}^C_{ext}(Z)\to\bC$.
\end{proposition}

\begin{proof}
In view of the previous proposition and of the properties of the Poincar\'e vacuum, it remains to be shown that $\omega$ is positive.
We shall check it for test functions $f$ and $f^\prime$ that can be factorized in the $z-$direction, namely of the form 
$f^{(\prime)}(\underline{x},z)=f^{(\prime)}_\perp(\underline{x})f^{(\prime)}_z(z)$
 where $f^{(\prime)}_\perp\in C^\infty_0(\mathbb{R}^3)$ and where $f^{(\prime)}_z\in C^{\infty}_0((0,d))$. Notice that, although we are not exhausting all possible elements of $C^\infty_0(\mathring{Z})$, we are still considering a dense subset, which suffices as far as positivity is concerned. With respect to this kind of functions we can introduce the following distribution on $C^{\infty}_0((0,d)\times (0,d))$
\[
 w^{f^\prime_\perp,f_\perp} (f^\prime_z,f_z) \doteq \widetilde{\omega}^0_2({f^\prime_\perp}f^\prime_z,f_\perp f_z) = \lim_{\epsilon\to 0^+}\int_0^d dz \int_0^d dz'  w_{2,\epsilon}^{f^\prime_\perp,f_\perp}(z-z') f_z(z) f^\prime_z(z'),
\]
where as usual the limits are meant in the weak sense. Since $w^{f^\prime_\perp,f_\perp}$ is a Schwartz distribution, see {\it e.g.} \eqref{eq:2pt-ft}, we might rewrite it in the Fourier domain 
\begin{equation}\label{eq:2p-ft-vacuum}
\widetilde{\omega}^0_2(\overline{f^\prime_\perp}{f^\prime_z},f_\perp f_z)  
=
 w^{f^\prime_\perp,f_\perp} (f^\prime_z,f_z)
 = \int_{\mathbb{R}}  d\xi\, \widehat{w}_2^{\overline{f^\prime_\perp},f_\perp}(\xi) {\widehat{f}^\prime_z(\xi)}\widehat{f}_z(\xi).
\end{equation}
Notice that, since the two-point function $\widetilde{\omega}_2$ of the Poincar\'e vacuum is itself a quadratic form, we have that $\widehat{w}^{\overline{f_\perp},f_\perp}_2(\xi)$ is a positive function which is continuous almost everywhere. In particular, from the expression of the spectrum built in \eqref{eq:2pt-ft}, we can infer that continuity could fail only at $\xi=0$, although $\widehat{w}^{\overline{f_\perp},f_\perp}_2(\xi)$ is a locally integrable function, also in a neighbourhood of $0$.

Let us now consider $w^{f_\perp,f_\perp}$ applied to $(\overline{f_z}, N f_z)$. By Poisson summation formula it holds $\sum_{l}f_z(z+2dl) = \sum_n f_n e^{i n z \pi/d}$  where $f_n$ are the Fourier coefficients of $f_z$ computed in the interval $[-d,d]$ and they coincide with the ordinary Fourier transform evaluated at $\xi=n \pi/d $, namely $f_n=\widehat{f}_z(n \pi/d)$. Hence, taking into account the anti-symmetrization present in $N$, $N f_z = \sum_{n} (f_n-f_{-n}) e^{i n z \pi/d}$. Furthermore, its Fourier transform can be computed in a distributional sense as
\[
\widehat{N f_z} := \left(\widehat{f}_z(\xi)-\widehat{f}_z(-\xi)\right) \sum_n \delta \left(\xi-n \frac{\pi}{d}\right).
\]

Dropping the superscripts $\overline{f_\perp},f_\perp$ from both $w$ and $\widehat{w}$ it holds
\[
w(\overline{f_z}, N f_z)
= \int d\xi \; {\widehat{w}_2}(\xi)\; \left(\widehat{f}_z(\xi)-\widehat{f}_z(-\xi)\right) \sum_n \delta \left(\xi-n \frac{\pi}{d}\right)   \overline{\widehat{f}_z(\xi)}.
\]
Notice that, despite of the presence of an infinite sum of Dirac delta functions, the previous expression is well defined because $\widehat{w}_2(\xi)$ is continuous for $\xi\neq 0$, it grows at most polynomially for large $|\xi|$ and it is 
bounded close to zero\footnote{In order to check boundedness of $\widehat{w}_2(\xi)$, notice
from \eqref{eq:2pt-ft} that for some positive constant $C$ 
\[
|\widehat{w}(\xi)| \leq C\sup_{t,t'\in I}\int_{\mathbb{R}^2} dk_\perp \frac{1}{\sqrt{k_\perp^2+\xi^2}}|\widehat{f}_\perp(t, k_\perp)||\widehat{f^\prime}_\perp(t', k_\perp)|
\leq C\sup_{t,t'\in I}\int_{0}^\infty d|k_\perp| \int_0^{2\pi} d\theta |\widehat{f}_\perp(t, k_\perp)||\widehat{f^\prime}_\perp(t', k_\perp)|
\]
where the supremum is taken in some interval $I$ chosen in such a way that $I\times\mathbb{R}^2$ contains the supports of both $f_\perp$ and $f^\prime_\perp$.
Furthermore, $\widehat{f}_\perp$ and $\widehat{f}^\prime_\perp$ are the spatial Fourier transform of $f_\perp$ and $f^\prime_\perp$
and hence they decay rapidly for large values of $|k_\perp|$. The result of the two integrals can thus be bounded by some positive constant.}.
The only delta function in the sum which could give a divergent contribution is the one supported at $0$. Since $\widehat{f}_z$ is a Schwartz function, $\left(\widehat{f}_z(\xi)-\widehat{f}_z(-\xi)\right)$ vanishes, however, at zero and hence, thanks to the boundedness of $\widehat{w}_2(\xi)$ near that point,  the contribution of the delta function supported at $0$ vanishes. We have
\begin{gather*}
w(\overline{f_z}, N f_z)= \sum_n \widehat{w}_2(n\pi/d) (f_n-f_{-n})  \overline{f_n}= \sum_{n\geq 1} \widehat{w}_2(n\pi/d) |f_n-f_{-n}|^2  
\end{gather*}

where, in the last equality, we use the fact that $\omega_2$ is symmetric under $z-$reflections and hence $\widehat{w}_2(n\pi/d) = \widehat{w}_2(-n\pi/d)$. The last term of the above chain of equalities is positive because it is a sum of positive quantities, since we have started from the two-point function of a state and, hence, $\widehat{w}_2(n\pi/d)$ is a quadratic form for every $n$.
\end{proof}

In order to generalize the result obtained for another quasi-free Hadamard state $\widetilde{\omega}$ whose two-point function integral kernel enjoys the symmetries stated in \eqref{eq:symmetry-casimir}, we recall that the two-point function of such state differs from the vacuum one by a smooth function $W(x,x^\prime)$. We have now to make sure that $\mathbb{I}\otimes N$ can be applied also to $W(x,x^\prime)$. To this end, we need to impose technical restrictions on the admissible class of smooth functions. 

\begin{proposition}\label{pr:image-method-casimir}
Let $\widetilde{\omega}$ be a quasi-free state of Hadamard form for $\mathcal{A}^{KG}(\bR^4)$. Suppose that the integral kernel of its two-point function $\widetilde{\omega}_2(x,x^\prime)=\widetilde{\omega}_2^0(x,x^\prime)+W(x,x^\prime)$ is invariant under \eqref{eq:symmetry-casimir}. Suppose that the following conditions hold for the smooth part $W\in C^\infty(\bR^4\times\bR^4)$:
\begin{enumerate}
\item[(i)] the function $W^{f_\perp,h_\perp}(z,z^\prime):= \int\limits_{\bR^3\times\bR^3}d^3\underline{x}d^3\underline{x}^\prime\; W(\underline{x},z,\underline{x}^\prime,z^\prime)  f_\perp(\underline{x}) f^\prime_\perp (\underline{x}^\prime)$ 
lies in $\mathcal{S}'(\mathbb{R}^2)$  for every $f_\perp,h_\perp\in \mathcal{D}(\mathbb{R}^3)$. 
\item[(ii)] for every value of $x^3$ and $x^{\prime 3}$, $W^{f_\perp,f^\prime_\perp}(z,z^\prime)$ generates a distribution in $\mathcal{D}'(\mathbb{R}^6)$, hence it is continuous on $\mathcal{D}(\mathbb{R}^6)$. 
\item[(iii)] Let $w(z-z^{\prime})\doteq W^{f_\perp,f^\prime_\perp}(z,z^{\prime})$ and let $\widehat{w}(\xi)$ be its Fourier transform. It is a continuous function for $\xi \geq \frac{\pi}{d}$, 
\item[(iv)] $\xi\mapsto \widehat{w}(\xi) \xi$  is a continuous function in a neighbourhood of $\xi=0$ and it vanishes for $\xi=0$.
\end{enumerate}
Hence, in view of Proposition \ref{surjective} we can extend $\widetilde\omega_2$ to a map on $C^\infty_0 (\mathring{Z})\times N[C^\infty_0(\mathring{Z})]$ and 
 \[
{\omega}_2(f^\prime,f) =\widetilde{\omega}_2(f^\prime,Nf) .
\]
gives rise to a quasi-free state $\omega:\mathcal{A}^{C}_{ext}(Z)\to\bC$.
\end{proposition}

\begin{proof}
Consider a compactly supported smooth function $f\in \mathcal{D}(\mathring{Z})$ which can be factorized in the following way $f(\underline{x},z)=f_\perp(\underline{x})f_z(z)$. Let us study $Nf$ and notice that $N$ acts only on $f_z$. Furthermore, by the Poisson summation formula (see \cite[\S 7.2]{Hormander1}), we know that $N f_z(z) = \sum_n (f_n-f_{-n}) e^{i n z \pi/d}$ and, as discussed in the proof of the previous proposition, the Fourier transform can be computed in the distributional sense yielding
\[
\widehat{N f_z}(\xi) := \left(\widehat{f}_z(\xi)-\widehat{f}_z(-\xi)\right) \sum_n \delta \left(\xi-n \frac{\pi}{d}\right)
\] 
where $\delta$ is the Dirac delta function. 
For every other $f^\prime\in\mathcal{D}(\mathring{Z})$ which can also be factorized, we analyze 
\begin{gather*}
W(f^\prime,Nf) := \int_{\mathbb{R}^4\times\mathbb{R}^4}d^4x d^4x^\prime\,  f^\prime(x) W(x,x^\prime) Nf(x^\prime)   = \int_{\mathbb{R}}d\xi\, \widehat{W}^{f^\prime_\perp,f_\perp}(\xi) \overline{\widehat{f}^\prime_z(\xi)} \widehat{N f_z}(\xi)  = \\
=\int_{\mathbb{R}} d\xi\, \widehat{w}(\xi)  \overline{\widehat{f}^\prime_z(\xi)} \left(\widehat{f}_z(\xi)-\widehat{f}_z(-\xi)\right) \sum_n \delta \left(\xi-n \frac{\pi}{d}\right).
\end{gather*}

The previous expression is well defined for the following reasons: 
 \begin{itemize}
\item[a)]  conditions $(iii)$ implies that  $\widehat{w}(\xi)$ is continuous for $|\xi| \geq \pi/d$, 
\item[b)] thanks to hypothesis $(i)$, $w(z)$ is a Schwartz distribution, hence its Fourier transform, grows at most polynomially for large $\xi$ and 
\item[c)]  requirement $(iv)$ implies that $\widehat{w}(\xi) \xi$ is continuous near zero and vanishes for $\xi=0$. \end{itemize}
Hence, the Dirac delta supported in $0$ gives a vanishing contribution to the sum because $\left(\widehat{f}_z(\xi)-\widehat{f}_z(-\xi)\right)/\xi$ is a continuous function near zero and hence $\widehat{w}(\xi) \xi \cdot \left(\widehat{f}_z(\xi)-\widehat{f}_z(-\xi)\right)/\xi$ is continuous in $0$ and there it vanishes. Furthermore, what remains is 
\[
W(f^\prime, Nf)
=
\sum_n  \widehat{w}\left( \frac{n\pi}{d}\right)  \overline{{f^\prime}_n} \left({f}_n-{f}_{-n}\right) 
= \sum_{n\geq 1}  \widehat{w}\left( \frac{n\pi}{d}\right) \overline{({f^\prime}_n - f^\prime_{-n})} \left({f}_n-{f}_{-n}\right) 
 \]
which is continuous with respect to the topology of $\mathcal{D}'((0,d)\times(0,d))$. Hence, taking into account hypothesis $(ii)$,  $W(f^\prime,Nf)$ is separately continuous on $\mathcal{D}((0,d)\times (0,d))\otimes \mathcal{D}(\mathbb{R}^6)$ and thus it is a distribution in $\mathcal{D}'(\mathring{Z}\times\mathring{Z})$.  

For this reason, ${\omega}_2$ is also a well-defined distribution being the sum of $\omega_2^0$ and $W\circ\left(\mathbb{I}\otimes N\right)$.
Positivity remains to be shown, but it can be checked following a proof similar to the proof of Proposition \ref{pr:vacuum-positivity}, hence we shall omit it.
\end{proof}

The requirements of the previous proposition are quite involved to check. For this reason, in the following lemma we give an alternative sufficient condition which implies the four points  assumed in the previous proposition.

\begin{lemma}\label{le:image-method-casimir}
Let $\widetilde{\omega}$ be a quasi-free state of Hadamard form for $\mathcal{A}^{KG}(\bR^4)$. Suppose that its two-point function $\widetilde{\omega}_2=\widetilde{\omega}_2^0$ is invariant under $z-$reflections and under $z-$translations as in \eqref{eq:symmetry-casimir}. 
Consider the smooth function $W:=\widetilde{\omega}_2-\widetilde{\omega}_2^0$. Suppose that the following conditions hold:
\[
W\in L^\infty(Z)  \;,\qquad   \frac{\partial}{\partial z}W(\underline{x},z,\underline{x}^\prime,z^\prime) \in L^1(\mathbb{R},dz)
\]
uniformly in $\underline x$ and $\underline x^\prime$, then the hypotheses of the previous proposition are satisfied and thus the following expression
\[
{\omega}_2(f^\prime,f) = \widetilde{\omega}_2(f^\prime,Nf) .
\]
is a well defined two-point function of a quasi-free state $\omega:\mathcal{A}^{C}_{ext}(Z)\to\bC$.
\end{lemma}
\begin{proof}
Since $W$ is bounded, it is the integral kernel of a Schwarz distribution.
Hence, by the Schwartz kernel theorem $W$ can be seen as a map between smooth functions over $\mathbb{R}^6$ and Schwartz distributions over $\mathbb{R}^2$.
The first three requirements of Proposition \ref{pr:image-method-casimir} descend immediately. The forth one requires a few words.
Since the derivative along $z$ of $W$ is in $L^1$, by the Riemann-Lebesgue lemma, 
its Fourier transform along the $z-$direction $\widehat{w}(\xi)$ is equal to a continuous function $u(\xi)$ divided by $\xi$. 
Furthermore, since $W$ is symmetric under reflections generated by $\iota_z$, $\widehat{w}$ must be invariant under mapping of $\xi\to-\xi$, and thus $u(\xi)=\xi\widehat{w}(\xi)$ is an odd continuous function, hence it must vanish for $\xi=0$.
\end{proof}

Before concluding this section we analyze the singular structure of Hadamard states obtained by the image method described so far. 
We already know that these states are of Hadamard form when restricted on globally hyperbolic sub regions of $\mathbb{H}$, hence therein the singular structure is known, however we expect further singularities when states for the full algebra $\mathcal{A}^C(Z)$ is considered. Actually, the following proposition holds.

\begin{proposition}
Consider the two-point function of a quasi-free state ${\omega}$ for $\mathcal{A}^{C}(Z)$ obtained by the image method starting from a quasi-free Hadamard state $\widetilde{\omega}$ of $\mathcal{A}^{KG}(\bR^4)$. 
The wave front set of its two-point function $\omega_2$ has the following form 
\[
WF({\omega}_2) = \left\{(x,x^\prime,k_x, -k_{x^\prime}) \in T^*\left(\mathring{Z}\times\mathring{Z}\right)\setminus\{0\} | (x,k_x)\sim_Z (x^\prime,k_{x^\prime}),   k_x\triangleright 0     \right\}
\]
where $(x,k_x)\sim_Z (x^\prime, k_{x^\prime})$ whenever there exists a null geodesic $\gamma$ reflected at the boundaries a countable number of times, such that $x,y$ are its end points, $k_x$ is the cotangent vector to $\gamma$ at $x$ while $k_y$ is the parallel transport of $k_x$ along $\gamma$. 
\end{proposition}
\begin{proof}

We recall that
\[
{\omega}_2(x,x^\prime) = \sum_{n\in\mathbb{N}} \left[ \widetilde{\omega}_2(x,(\underline{x}^\prime,z^\prime + 2nd)) -   \widetilde{\omega}_2(x,(\underline{x}^\prime,-z^\prime + 2nd)) \right],
\]
Hence, $WF({\omega}_2)$ is contained in the union of the wavefront sets of $\widetilde{\omega}_2(x,(\underline{x}^\prime,z^\prime + 2nd))$ and of $\widetilde{\omega}_2(x,(\underline{x}^\prime,-z^\prime + 2nd))$. 

Let us analyze $WF(\widetilde{\omega}_2(x,(\underline{x}^\prime,z^\prime + 2nd)))$. 
Notice that $\widetilde{\omega}_2(x,(\underline x^\prime,z^\prime+2nd))$ is nothing but as $\widetilde{\omega}_2$ in Minkowski with a translation applied to $x^\prime$. Hence we just need to apply the corresponding transformation on its wavefront set to obtain the wavefront set of $WF(\widetilde{\omega}_2(x,(\underline{x}^\prime,z^\prime + 2nd)))$. Furthermore, if the points $(x,x^\prime)$ are contained in its singular support, this means that $x$ and $\iota_{2nd}(x^\prime)$ are connected by a null geodesic in Minkowski spacetime.
This geodesic in Minkowski passes trough the points $z$ where $z^3$ is a multiple of $d$, $|2n|$ times. Hence, in the Casimir region, it is like a null geodesic reflected $2n$ times at the boundaries. 
We can treat in a similar way $WF(\widetilde{\omega}(x,(\underline{x},-z^\prime + 2nd)))$ and it coincides with the wave front set of $\widetilde{\omega}$ where the second entry of that distribution is reflected and translated $2n$ times. Hence, $(x,x^\prime)$ are in its singular support only if they are connected by a null geodesic reflected $|2n-1|$ times at the boundaries.

Finally, we notice that the wave front set of $\widetilde{\omega}(x,(\underline{x}^\prime,z^\prime + 2nd))$ and of $\widetilde{\omega}(x,(\underline{x}^\prime,-z^\prime + 2nd))$ are all disjoined, (their singular support might overlap only when both $z=z^\prime=d/2$ but in this case the corresponding covectors have opposite $z-$direction). Hence, in the sum defining ${\omega}_2$ no cancellation of singularity might occur. We thus conclude that $WF({\omega}_2)$ coincides with the union of the wave front sets of the distributions in the sum written above.
\end{proof}

\subsection{The vacuum and the KMS states for the Casimir system}\label{lastsec}

In this subsection, we shall construct states $\omega^T:\mathcal{A}^C_{ext}(Z)\to\bC$ at finite temperature $T$ for the Casimir system. We shall show that these states are obtained applying the image method to a KMS state for a Klein-Gordon field on Minkowski spacetime. As a corollary, we obtain that $\omega^0$ is the vacuum state of the theory and it coincides with $\lim_{T\to 0} \omega^T$. Our computations are consistent with the literature on the topic, see for example \cite{Brown:1969na, Fulling, Kennedy:1979ar} for the thermal case and \cite{Fulling:1989nb} for the vacuum. 

As before, we work at the level of two-point function. Hence, let us suppose that the hypotheses of Proposition \ref{pr:image-method-casimir} are met. If so, we can apply the image method to a state $\widetilde{\omega}$ on $\mathcal{A}^{KG}(\bR^4)$ to obtain a quasi-free Hadamard state ${\omega}$ for $\mathcal{A}^C_{ext}(Z)$, such that $\omega_2(f,f^\prime) = \widetilde{\omega}_2(f,Nf^\prime)$, $f,f^\prime\in C^\infty_{0,C}(\bR^4)$.
Suppose also that the state $\widetilde{\omega}$ is invariant under the natural action induced on it by the time translation $t_\xi$ of step $\xi\in\bR$. Since $N$ commutes with $t_\xi$, also the state $\omega$ must be invariant under time translations.

Consider now the quasi-free KMS state $\widetilde{\omega}^T:\mathcal{A}^{KG}(\bR^4)\to\bC$ at temperature $T$ which is invariant under the action induced by $t_\xi$. For every $f,f^\prime\in C^\infty_0(\bR^4)$  the function 
$\xi\mapsto \omega_2(t_\xi f  , g)$ is analytic in the strip $\Im(\xi) \in [0,\beta]$ where $\beta = (k_B T)^{-1}$ is the inverse temperature and $k_B$ is the Boltzmann constant. Furthermore, the KMS condition holds, namely 
\[
\widetilde{\omega}^T_2(t_{i\beta} f  , f^\prime) = \widetilde{\omega}^T_2(f^\prime,f).
\]
We recall also that
\[
\widetilde{\omega}^{T}_2(x,x^\prime) = 
\lim_{\epsilon\to 0^+} \frac{1}{2\pi \beta |{\bf x}-{\bf x}^\prime|}
\frac{\sinh{\left(2\pi \frac{|{\bf x}-{\bf x}^\prime|}{\beta}\right)}}{\cosh{\left(2\pi \frac{|{\bf x}-{\bf x}^\prime|}{\beta}\right)}-\cosh{\left(2\pi \frac{(x^0-x^{\prime 0}-i\epsilon)}{\beta}\right)}}
\;, \Im(x^0-x^{\prime 0})\in(-\beta+\epsilon,0]\;,
\]
where we use $x^0$ for the time coordinate and $\bf x$ for the space coordinates. Furthermore 
\[
\widetilde{\omega}_2^T(x,x^\prime) = \lim_{\epsilon\to 0^+ }\int_{\bR^3} d^3{\bf k} 
\frac{e^{i {\bf k} \cdot ({\bf x} -{\bf x}^\prime)}}{2{|{\bf k}|}}  
\left(
\frac{e^{-i |{\bf k}| (x^0-x^{\prime 0})}}{1- e^{-\beta|{\bf k}|}} +\frac{ e^{i|{\bf k}|(x^0-x^{\prime 0})}}{ e^{\beta|{\bf k}|}-1}
\right) e^{-\epsilon|{\bf k}|}
\]
We shall check that, it is possible to apply the image method to this state by analyzing the behavior of $W := \widetilde{\omega}_2^T-\widetilde{\omega}_2^0$
and verifying that the hypotheses of Lemma \ref{le:image-method-casimir}
is satisfied and thus Proposition \ref{pr:image-method-casimir} holds. First of all, we notice that $W$ is a Schwartz distribution, 
which has the desired symmetry properties \eqref{eq:symmetry-casimir}. The spatial Fourier transform of its integral kernel has the following form 
\[
\widehat{W}(x^0,x^{\prime 0};{\bf k}) = C\frac{1}{{|{\bf k}|}}  
\left(
\frac{\cos(|{\bf k}|(x^0-x^{\prime 0})) }{ e^{\beta|{\bf k}|}-1}
\right).
\]
It is a smooth function except when $|{\bf k}| = 0$ and it decays rapidly for large $|{\bf k}|$. From this observation conditions
 $(i)$,$(ii)$ and $(iii)$ of Proposition \ref{pr:image-method-casimir} are met. It remains to prove the $(iv)$. In order to check it we proceed  analyzing 
\[
\widehat{w}^T(\xi) = \int_\bR dt \int_\bR dt' \int_{\mathbb{R}^2} dk_\perp 
\widehat{W}(t,t';k_\perp,\xi) \widehat{f}_\perp(t, k_\perp) \overline{\widehat{f}^\prime_\perp(t', k_\perp)}.
\]
for a pair of compactly supported function $f_\perp,f^\prime_\perp\in\mathcal{D}(\mathbb{R}^3)$. Above, $\widehat{f}_\perp(t, k_\perp)$ is the spatial (two-dimensional) Fourier transform of $f_\perp(t,\underline x^1,\underline x^2)$.
Notice that there exists a positive constant $C$ such that  $|\widehat{W}(t,t';{\bf k})| \leq C/|{\bf k}|^2$. Hence 
\[
|\widehat{w}^T(\xi)| \leq C\sup_{t,t'\in I}\int_{\mathbb{R}^2} dk_\perp \frac{1}{k_\perp^2+\xi^2}|\widehat{f}_\perp(t, k_\perp)||\widehat{h}_\perp(t', k_\perp)|,
\]
where the supremum is taken in an interval $I$ chosen in accordance to the supports of both $f_\perp$ and $f^\prime_\perp$.
Since, $\widehat{f}_\perp$ and $\widehat{f}^\prime_\perp$ are two Schwartz functions it holds that 
 \[
|\widehat{w}(\xi)| \leq \sup_{t,t'\in I}C'(t,t')\int_{0}^\infty dk \frac{k}{k^2+\xi^2}\frac{1}{1+ k^2}
\]
for some positive set of constants $C'(t,t')$ bounded in $I^2$. The $k-$integral can be computed and it yields a function of $\xi$ which is logarithmically divergent near $0$, and hence, also requirement $(iv)$ of Proposition \ref{pr:image-method-casimir} is met.

\bigskip 
For completeness we check the applicability of the image method directly on the two-point function. We obtain
\begin{gather}
{\omega}^T_{2}(x,x^\prime)\doteq \left(\widetilde{\omega}^{T}_2\left(\mathbb{I}\otimes N\right)\right)(x,x^\prime)=
\notag
\\
-\sum\limits_{n=-\infty}^\infty\left(\frac{1}{2\pi\beta r_n}\frac{\sinh\frac{2\pi r_n}{\beta}}{\cosh\frac{2\pi r_n}{\beta}-\cos\frac{2\pi i}{\beta}(\underline{x}^0-\underline{x}^{\prime 0}+i\epsilon)}
-
\frac{1}{2\pi\beta \widetilde{r}_n}\frac{\sinh\frac{2\pi \widetilde{r}_n}{\beta}}{\cosh\frac{2\pi \widetilde{r}_n}{\beta}-\cos\frac{2\pi i}{\beta}(\underline{x}^0-\underline{x}^{\prime 0}+i\epsilon)}\right),
\label{eq:2pt-kms-casimir}
\end{gather}
where $r^2_n\doteq (\underline{x}^1-\underline{x}^{\prime 1})^2+(\underline{x}^2-\underline{x}^{\prime 2})^2+(z-z^\prime+2nd)^2$ while $\widetilde{r}^2_n\doteq (\underline{x}^1-\underline{x}^{\prime 1})^2+(\underline{x}^2-\underline{x}^{\prime 2})^2+(-z-z^\prime+2nd)^2$.
Notice that, for every $\epsilon>0$ and for every $x,y$ in $Z$ we have the the sum is absolutely convergent. 
As a matter of facts, for large $n$, both $r_n$ and $\widetilde{r}_n$ grow like $2nd$ hence, the asymptotic behavior of the $n-$th element of series is 
governed by 
\begin{gather*}
\frac{1}{2\pi\beta r_n} -\frac{1}{2\pi\beta \widetilde{r}_n} =  
\frac{1}{2\pi\beta } \frac{\widetilde{r}_n-r_n}{r_n\widetilde{r}_n} =
\frac{1}{2\pi\beta } \frac{\widetilde{r}_n^2-r_n^2}{r_n\widetilde{r}_n (r_n+\widetilde{r}_n)}
\end{gather*}
and the right hand side of the previous expression is majored by $C/n^2$ hence it can be summed.

\bigskip
We conclude this section with a proposition which ensures that the image method preserves the thermal properties of states.

\begin{proposition}
The quasi-free state $\omega^T:\mathcal{A}^C_{ext}(Z)\to\bC$, whose two-point function
$\omega_2^T$  is obtained applying the image method to the two-point function $\widetilde{\omega}_2^T$ of the KMS state $\widetilde{\omega}^T$ as in \eqref{eq:2pt-kms-casimir} is a KMS state. The limit of $\omega^T$ as $T\to 0$ is a vacuum state.  
\end{proposition}

\begin{proof}
In order to prove the proposition, we want now to show that ${\omega}^T_2(f,f^\prime) = \widetilde{\omega}^T_2(f,N f^\prime)$ for $f,f^\prime\in C^\infty_0(\mathring{Z})$, enjoys the KMS condition in $\mathcal{A}^C_{ext}(Z)$. To this end we recall that the KMS condition can alternatively be written as 
\[
\widetilde{\omega}^T_2(t_{i\beta}(f),f^\prime) - \widetilde{\omega}^T_2(f,f^\prime) =   - i E(f,f^\prime) 
\]
where $E$ is the causal propagator of the theory. 
Hence, let us analyze it for ${\omega}^T$
\[
{\omega}^T_2(t_{i\beta}(f),f^\prime) - {\omega}^T_2(f,f^\prime) = 
\widetilde{\omega}^T_2(t_{i\beta}(f),N f^\prime) - \widetilde{\omega}^T_2(f,N f^\prime) =   - i E(f,N f^\prime) = - i E_Z(f,f^\prime),
\]
where $E_Z$ is constructed in Definition \ref{Casimirpropagator}
Since in the limit $\beta\to 0$ we recover ${\omega}_2^0$ we might safely say that ${\omega}^0$ is the ground state of the Casimir system.
\end{proof}

Notice that the very same conclusion could have been drawn using instead a more general argument following the analysis of \cite{Sahlmann:2000fh}. It is noteworthy that the analysis of this section could have been performed for the Hadamard states of a massive real scalar field on the whole Minkowski spacetime. Yet, in such case, on account of the fall-off properties at infinity of the Poincar\'e vacuum, we would have obtained far better convergence results of the image method.

\subsection{Wick ordering in a Casimir system} 

\noindent To conclude the section, as for a Casimir-Polder system we want to make contact with the standard results in the literature concerning the expectation value of the regularized two-point function and stress-energy tensor (see \cite{Sopova&Ford:2002}). To this avail we need first of all to define the extended algebra of Wick polynomials. The procedure is identical to the one discussed in Section \ref{Wick1} and, thus, we will not repeat it here. Recall that the main outcome was the possibility to introduce an algebra of extended observables on globally hyperbolic submanifolds $O$.
Furthermore, thereon, $\mathcal{A}^{C}_\mu(O)$ is $*$-isomorphic (actually it coincides) to $\mathcal{A}^{KG}_\mu(O)$.  For the same reasons discussed in the Casimir-Polder case, however, the extended algebras $\mathcal{A}^{C}_\mu(O)$ can be realized as part of a global extended algebra $\mathcal{A}^C_{ext}(Z)$ only after a deformation of the $\star_Z$ product into a globally defined one. This  can be built for example by replacing $H$ with the two-point function of a Hadamard state . 

Despite of this difficulty, we can locally make sense to observables like the stress tensor or the Wick square, and in particular, we have:
\begin{proposition}\label{Casimirphys}
Let us consider a massless, conformally coupled real scalar field and let $\omega:\mathcal{A}^C_{ext}(Z)\to\bC$ be the quasi-free state whose two-point function $\omega_2=\left(N\otimes\mathbb{I}\right)\omega_{2,V}$ is built with the image method from the Poincar\'e vacuum. Then, for all $\zeta\in C^\infty_0(Z)$,
$$
\omega^0(\normal{\widehat{\phi}^2}_{H}(\zeta))=\int\limits_{\bR^4}d^4x \frac{\zeta(x)}{48d^2}\left(1-\frac{3}{\sin^2\frac{\pi z}{d}}\right),$$
and
$$
\omega^0(\normal{\widehat{T_{\mu\nu}}}_{H}(\zeta))=A^\prime_{\mu\nu}\frac{-1}{1440 d^4}\int\limits_{\bR^4}d^4x\,\zeta(x)\left[ 1 + (6\xi-1)\frac{5\pi^2}{2}\left( \psi^{(3)}\left(1-\tfrac{z}{d}\right) - \psi^{(3)}\left(\tfrac{z}{d}\right) \right)\right],$$
where $\{T_{\mu\nu}\}$ are the components of the stress-energy tensor \eqref{improved} and $\psi(x)$ is the logarithmic derivative of the gamma function. Furthermore $A^\prime$ is the matrix $\mathrm{diag}(-1,1,1,3)$
\end{proposition}

\begin{proof}

Recall that, according to Proposition \ref{pr:vacuum-positivity}, $\omega$ is a Hadamard state as per Definition \ref{def:hadamard-casimir}. In order to compute  the Wick squared scalar field, we recall result of the image method and we obtain 
\begin{gather*}
\omega^0(\normal{\widehat{\phi}^2}_{H}(\zeta))=\\
\sum\limits_{\substack{n=-\infty\\ n\neq 0}}^\infty\int\limits_{\bR^4\times\bR^4}\!\!d^4xd^4x^\prime\left(\widetilde\omega^0_{2}(\underline{x}-\underline{x}^\prime,z-z^\prime+2nd)-\widetilde\omega^0_{2}(\underline{x}-\underline{x}^\prime,-z-z^\prime+2nd)\right)\zeta(\underline{x},z)\delta(x-x^\prime)=\\
=\frac{1}{4\pi^2}\sum\limits_{\substack{n=-\infty\\ n\neq 0}}^\infty \int\limits_{\bR^4}d^4x\,\left(\frac{1}{(2nd)^2}-\frac{1}{(2z+2nd)^2}\right)\zeta(\underline{x},z)=\int\limits_{\bR^4}d^4x\left(\frac{1}{96d^4}-\frac{1}{32d^4}\frac{1}{\sin^2\frac{\pi z}{d}}\right)\zeta(\underline{x},z),
\end{gather*}
where we used the smoothness property of the sum of the integral kernels in the region of interest, first to deduce that the result of the integrals is finite and then to exchange the sum with the integrals. In the last equality we have computed the sum by using for the first term the definition of the Riemann zeta function and in the second still \cite[\S 1.445]{tableofintegrals}.
In order to compute the expectation value of the smeared Wick ordered time-diagonal component of the stress-energy tensor, we follow the same procedure as in the proof of Lemma \ref{reale}, that is
$$
\omega(\normal{\widehat{T_{\mu\nu}}}_{H}(\zeta))=\int\limits_{\bR^4\times\bR^4}d^4x\,d^4x^\prime D^{(x,x^\prime)}_{\mu\nu}\left(\omega_2(x,x^\prime)-H(x,x^\prime)\right)\zeta(x)\delta(x-x^\prime),
$$
where $D^{(x,x^\prime)}_{\mu\nu}$ is the same as in \eqref{D00}. Following the same procedure as for $\omega(\normal{\widehat{\phi}^2}_{H}(\zeta))$ the sought result descends.
\end{proof}

\section*{Acknowledgments}  
The authors would like to thank M.~Benini, N.~Drago, C.~Fewster, T.~-P.~Hack, I.~Kahvkine,  V.~Moretti and D.~Siemssen for enlightening discussions. We are especially grateful to K.~Fredenhagen for critical and enlightening discussions especially on Section \ref{lastsec}.\\


\begin{thebibliography}{999}

\bibitem[Bar13]{Baernew}
  C.~B\"ar,
  {\it Green-hyperbolic operators on globally hyperbolic spacetimes},
  [arXiv:1310.0738 [math-ph]], to appear on Comm. Math. Phys.  
  
\bibitem[BDF09]{Brunetti:2009qc}
  R.~Brunetti, M.~Duetsch and K.~Fredenhagen,
  {\it ``Perturbative Algebraic Quantum Field Theory and the Renormalization Groups,''}
  Adv.\ Theor.\ Math.\ Phys.\  {\bf 13} (2009) 1541
  [arXiv:0901.2038 [math-ph]].  

\bibitem[BFR12]{Brunetti:2012ar}
  R.~Brunetti, K.~Fredenhagen and P.~L.~Ribeiro,
  {\it ``Algebraic Structure of Classical Field Theory I: Kinematics and Linearized Dynamics for Real Scalar Fields,''}
  arXiv:1209.2148 [math-ph].

\bibitem[BDH13]{Benini:2013fia}
  M.~Benini, C.~Dappiaggi and T.~-P.~Hack,
  {\it ``Quantum Field Theory on Curved Backgrounds -- A Primer''},
  Int.\ J.\ Mod.\ Phys.\ A {\bf 28} 1330023 (2013)
  [arXiv:1306.0527 [gr-qc]].  
  
\bibitem[BDM14]{Benini:2014rya}
  M.~Benini, C.~Dappiaggi and S.~Murro,
  {\it ``Radiative observables for linearized gravity on asymptotically flat spacetimes and their boundary induced states''},
  J. Math. Phys. {\bf 55}, 082301 (2014), arXiv:1404.4551 [gr-qc].  

\bibitem[BDS12]{Benini:2012vi} 
  M.~Benini, C.~Dappiaggi and A.~Schenkel,
  {\it``Quantum field theory on affine bundles''},
  Annales Henri Poincare {\bf 15} 171 (2014) 
  [arXiv:1210.3457 [math-ph]].  

\bibitem[BDS13]{Benini:2013tra}
  M.~Benini, C.~Dappiaggi and A.~Schenkel,
  {\it``Quantized Abelian principal connections on Lorentzian manifolds''},
   Commun.\ Math.\ Phys.\  {\bf 330} 123 (2014)
  [arXiv:1303.2515 [math-ph]].
 
\bibitem[Ben14]{Benini}
  M.~Benini,
  {\it ``Optimal space of linear classical observables for Maxwell $k$-forms via spacelike and timelike compact de Rham cohomologies''},
  arXiv:1401.7563 [math-ph].  
  
\bibitem[BF09]{Brunetti:2009pn}
  R.~Brunetti and K.~Fredenhagen,
  {\it ``Quantum Field Theory on Curved Backgrounds,''}
  arXiv:0901.2063 [gr-qc].   

\bibitem[BFV]{Brunetti:2001dx} 
  R.~Brunetti, K.~Fredenhagen and R.~Verch,
  {\it ``The Generally covariant locality principle: A New paradigm for local quantum field theory''},
  {Commun.\ Math.\ Phys.\ }  {\bf 237}, 31 (2003),
  arXiv:math-ph/0112041.  
 
\bibitem[BGP07]{BGP}
  C.~B\"ar, N.~Ginoux and F.~Pf\"affle,
  {\it Wave Equations on Lorentzian Manifolds and Quantization},
  1st edn. (Eur. Math. Soc., Z\"urich, 2007).  
  
\bibitem[BM69]{Brown:1969na}
  L.~S.~Brown and G.~J.~Maclay,
  {\it ``Vacuum stress between conducting plates: An Image solution,''}
  Phys.\ Rev.\  {\bf 184} (1969) 1272.   
 
\bibitem[Cas48]{Casimir:1948dh}
  H.~B.~G.~Casimir,
  {\it ``On the Attraction Between Two Perfectly Conducting Plates''},
  Indag.\ Math.\  {\bf 10} 261 (1948).

\bibitem[CCJ70]{Callan:1970ze}
  C.~G.~Callan, Jr., S.~R.~Coleman and R.~Jackiw,
  {\it ``A New improved energy - momentum tensor''},
  Annals Phys.\  {\bf 59} 42 (1970).   

\bibitem[CP48]{CP}
  H.~B.~G.~Casimir and D.~Polder,
  {\it ``The Influence of retardation on the London-van der Waals forces''},
  Phys.\ Rev.\  {\bf 73} 360 (1948).  
  
\bibitem[DC78]{Deutsch:1978sc}
  D.~Deutsch and P.~Candelas,
  {\it ``Boundary Effects in Quantum Field Theory''},
  Phys.\ Rev.\ D {\bf 20} 3063 (1979).  
  
\bibitem[Dim80]{Dimock}
  J.~Dimock,
  {\it ``Algebras of Local Observables on a Manifold''},
  {Commun.\ Math.\ Phys.\ } {\bf 77}, 219 (1980).  
  
\bibitem[FH95]{Fewster:1995bu}
  C.~J.~Fewster and A.~Higuchi,
  {\it ``Quantum field theory on certain nonglobally hyperbolic space-times,''}
  Class.\ Quant.\ Grav.\  {\bf 13} (1996) 51
  [gr-qc/9508051].  

\bibitem[FH12]{Fewster:2012bj} 
  C.~J.~Fewster and D.~S.~Hunt,
  {\it ``Quantization of linearized gravity in cosmological vacuum spacetimes''},
  Rev.\ Math.\ Phys.\  {\bf 25} 1330003 (2013), arXiv:1203.0261 [math-ph].

\bibitem[FePf03]{Fewster:2003ey} 
  C.~J.~Fewster and M.~J.~Pfenning,
  {\it ``A Quantum weak energy inequality for spin one fields in curved space-time''},
  J.\ Math.\ Phys.\  {\bf 44}, 4480 (2003),
  arXiv:gr-qc/0303106.
  
  \bibitem[FR12]{Fredenhagen:2012sb}
  K.~Fredenhagen and K.~Rejzner,
  {\it ``Perturbative algebraic quantum field theory,''}
  arXiv:1208.1428 [math-ph].
  
  \bibitem[FR87]{Fulling}
  S.~A.~Fulling and S.~N.~M.~Reijsenaars,
  {\it ``Temperature, periodicity and horizons''}
  Phys. Rep. {\bf 152}, 135 (1987).  
  
  \bibitem[Ful89]{Fulling:1989nb}
  S.~A.~Fulling,
  {\it ``Aspects of Quantum Field Theory in Curved Space-time,''}
  London Math.\ Soc.\ Student Texts {\bf 17} (1989) 1.
  
\bibitem[GR07]{tableofintegrals}
  I.~S.~Gradshteyn and I.~M.~Ryzhik
  {\it Tables of Integrals, Series and Products}, 
  7th edition, Academic Press (2007).  
 
\bibitem[HK63]{Haag:1963dh}
  R.~Haag and D.~Kastler,
  {\it ``An Algebraic approach to quantum field theory''},
  J.\ Math.\ Phys.\  {\bf 5}, 848 (1964).
   
\bibitem[Hac10]{Hack:2010iw} 
  T.-P.~Hack,
  {\it On the backreaction of scalar and spinor quantum fields in curved spacetimes},
  arXiv:1008.1776 [gr-qc], PhD thesis, Universit\"at Hamburg (2010).  

\bibitem[Her04]{Herdegen:2004xk}
  A.~Herdegen,
  {\it ``Quantum backreaction (Casimir) effect. I. What are admissible idealizations?''},
  Annales Henri Poincare {\bf 6} 657 (2005) 
  [hep-th/0412132].
  
\bibitem[Her05]{Herdegen:2005is}
  A.~Herdegen,
  {\it ``Quantum backreaction (Casimir) effect. II. Scalar and electromagnetic fields''},
  Annales Henri Poincare {\bf 7} 253 (2006)
  [hep-th/0507023].  
  
\bibitem[Her10]{Herdegen:2010rt}
  A.~Herdegen and M.~Stopa,
  {\it ``Global versus local Casimir effect''},
  Annales Henri Poincare {\bf 11} 1171 (2010)
  [arXiv:1007.2139 [hep-th]].  

\bibitem[HW01]{Hollands:2001nf} 
  S.~Hollands and R.~M.~Wald,
  {\it ``Local Wick polynomials and time ordered products of quantum fields in curved space-time''},
  Commun.\ Math.\ Phys.\  {\bf 223}, 289 (2001),
  arXiv:gr-qc/0103074.
  
\bibitem[HW14]{Hollands:2014eia}
  S.~Hollands and R.~M.~Wald,
  {\it ``Quantum fields in curved spacetime''},
  arXiv:1401.2026 [gr-qc].  
  
\bibitem[H\"or90]{Hormander1} 
  L.~H\"ormander, 
  {\it The Analysis of Linear Partial Differential Operators I},
  2nd edn. Springer, Berlin (1990).  

\bibitem[Kay78]{Kay:1978zr}
  B.~S.~Kay,
  {\it ``Casimir effect in quantum field theory''},
  Phys.\ Rev.\ D {\bf 20} (1979) 3052.
  
\bibitem[Kay92]{Kay:1992es}
  B.~S.~Kay,
  {\it ``The Principle of locality and quantum field theory on (nonglobally hyperbolic) curved space-times,''}
  Rev.\ Math.\ Phys.\ SI {\bf 1} (1992) 167.  
  
\bibitem[KCD79]{Kennedy:1979ar}
  G.~Kennedy, R.~Critchley and J.~S.~Dowker,
  {\it ``Finite Temperature Field Theory with Boundaries: Stress Tensor and Surface Action Renormalization,''}
  Annals Phys.\  {\bf 125} (1980) 346.  
  
\bibitem[Kuh05]{Kuhn}
  H.~K\"uhn,
  {\it Thermische Observablen gekoppelter Felder in Casimir-Effekt,}  
  Diploma Thesis in German, Universit\"at Hamburg (2005), available at \url{www.desy.de/uni-th/lqp/psfiles/dipl-kuehn.ps.gz}

\bibitem[Lee00]{Lee}
  J.~M.~Lee,
  {\it Introduction to smooth manifolds},
  Springer (2000).

\bibitem[Mil01]{Milton} K.~A.~Milton, 
  {\it The Casimir Effect: Physical Manifestations of Zero-point Energy.}
  World Scientific, (2001).

\bibitem[Mor03]{Moretti:2001qh} 
  V.~Moretti,
  {\it ``Comments on the stress energy tensor operator in curved space-time''},
  Commun.\ Math.\ Phys.\  {\bf 232}, 189 (2003),
  arXiv:gr-qc/0109048. 

\bibitem[Nie09]{Ole}
  O. ~Niekerken
  {\it Quantum and Classical Vacuum Forces at Zero and Finite Temperature}, Diploma Thesis in German (2009), Universit\"at Hamburg,
  available at \url{http://www-library.desy.de/preparch/desy/thesis/desy-thesis-09-019.pdf}

\bibitem[Rad96a]{Radzikowski:1996pa} 
  M.~J.~Radzikowski,
  {\it ``Micro-local approach to the Hadamard condition in quantum field theory on curved space-time''},
  Commun.\ Math.\ Phys.\ {\bf 179}, 529 (1996).

\bibitem[Rad96b]{Radzikowski:1996ei} 
  M.~J.~Radzikowski,
  {\it ``A Local to global singularity theorem for quantum field theory on curved space-time''},
  Commun.\ Math.\ Phys.\  {\bf 180}, 1 (1996).  
    
\bibitem[SV00]{Sahlmann:2000fh}
  H.~Sahlmann and R.~Verch,
  {\it ``Passivity and microlocal spectrum condition,''}
  Commun.\ Math.\ Phys.\  {\bf 214} (2000) 705
  [math-ph/0002021].

\bibitem[SDH12]{Sanders:2012sf}
  K.~Sanders, C.~Dappiaggi and T.~-P.~Hack,
  {\it ``Electromagnetism, Local Covariance, the Aharonov-Bohm Effect and Gauss' Law''},
  Commun.\ Math.\ Phys.\  {\bf 328} 625 (2014)
  [arXiv:1211.6420 [math-ph]].

\bibitem[Som06]{Sommer}
  C.~Sommer,
  {\it Algebraische Charakterisierung von Randbedingungen in der Quantenfeldtheorie ,}  
  Diploma Thesis in German, Universit\"at Hamburg (2006), available at \url{http://www.desy.de/uni-th/lqp/psfiles/dipl-sommer.ps.gz}  

\bibitem[SF02]{Sopova&Ford:2002}
  V.~Sopova and L.~H.~Ford,
  {\it ``Energy density in the Casimir effect''},
  Phys.\ Rev.\ D {\bf 66}, 045026 (2002).

\bibitem[Wal84]{Wald}
  R.~M.~Wald,
  {\it General Relativity}, 1st edn. 
  The University of Chicago Press, Chicago, (1984).
  
\bibitem[Wal94]{Wald2}
  R.~M.~Wald,
  {\it Quantum Field Theory in Curved Spacetime and Black Hole Thermodynamics}, 
  1st edn. The University of Chicago Press Chicago, (1994).

  
\end{thebibliography}
\end{document}